\documentclass[12pt]{article}

\usepackage{amssymb,amsmath}
\usepackage{graphicx,psfrag,epsf}
\usepackage{caption}
\usepackage{enumerate}
\usepackage{natbib}
\usepackage{url} %
\usepackage{algorithm}
\usepackage{algorithmic}
\usepackage{authblk}
\usepackage[utf8]{inputenc} %
\usepackage[T1]{fontenc}    %
\usepackage{url}            %
\usepackage{booktabs}       %
\usepackage{amsfonts}       %
\usepackage{nicefrac}       %
\usepackage{microtype}      %
\usepackage{lipsum}
\usepackage{fancyhdr}       %
\usepackage{subfig}
\usepackage{float}
\usepackage{amsfonts,amssymb,graphics,epsfig,verbatim,latexsym,url,amsbsy}
\usepackage{listings} %
\usepackage{empheq}
\usepackage{mathptmx}
\usepackage{bm}
\usepackage{helvet}
\usepackage{mathtools}
\usepackage{courier}
\usepackage{imakeidx} \makeindex[options = -s svind]
\usepackage{multicol}
\usepackage[bottom]{footmisc}
\usepackage{breqn}
\usepackage{color}
\usepackage{tikz}
\usepackage[toc,page]{appendix}
\usepackage{pgf}
\usepackage{bbm}
\usepackage{subcaption}
\usepackage{caption}
\usepackage{enumitem}
\usepackage{ntheorem}%
\usepackage{amsmath}   %
\usepackage{amssymb}   %
\usepackage{booktabs}  %
\usepackage{graphicx}  %

\graphicspath{{media/}}     %
\RequirePackage[colorlinks,citecolor=blue,linkcolor=blue,urlcolor=blue,pagebackref]{hyperref}

\newtheorem{assumption}{Assumption}
\newtheorem{theorem}{Theorem}
\newtheorem{lemma}{Lemma}
\newtheorem{proposition}{Proposition}

\newtheorem{proof}{Proof}
\newtheorem{remark}{Remark}

\newtheorem*{example_un}{Example}
\providecommand{\keywords}[1]
{
  \small	
  \textbf{\textit{Keywords---}} #1
}

\font\myfont=cmr12 at 11pt
  
\title{Weak instrumental variables due to ignored nonlinearities
 in  panel data: An Orthogonal Super Learner Control Function estimator}

\author[$\dagger$]{Monika Avila M\'{a}rquez
 \thanks{I am indebted to Prof. Rajen Shah for his invaluable feedback, as well as to
  Prof. Juan-Manuel Rodriguez-Poo, and Prof. Stefan Sperlich. 
   This paper has benefited from my time at the School of Economics, 
   University of Bristol. I thank the members of the Bristol
   Econometrics Group for their constructive comments. All errors are my own.
   } 

   }

\affil[$\dagger$]{ \myfont University of Geneva, Methods and Data Analysis\\
\myfont monika.avila@unige.ch}

\begin{document}
\maketitle

\begin{abstract}
A triangular structural panel data model with additively separable 
individual-specific effects is used to model the causal effect of a 
covariate on an outcome variable when there are unobservable confounders
 with some of them time-invariant.  In this setup, a linear specification 
 for the reduced-form equation might be problematic when the conditional 
 mean of the endogenous 
 covariate and the instrumental variables is nonlinear in the population.  
 The reason is that 
 ignoring the nonlinearity could lead to weak instruments (instruments are 
 weakly correlated with the endogenous covariate) due to misspecification
 as shown using a generalized concentration parameter for panel data.
 As a solution,  we propose a triangular simultaneous equation model 
 for panel data with additive separable individual-specific fixed effects 
 composed of a linear structural equation with a nonlinear reduced form equation.
  The parameter of interest is the structural parameter of the endogenous variable. 
   The identification of this parameter is obtained under the assumption of 
   available exclusion restrictions and using a control function approach. We derive a Neyman orthogonal score, and  
     we provide an estimator that 
     we call Orthogonal Super Learner Control Function (Orthogonal SLCF) estimator.  
     The estimation procedure is composed of two main steps 
     and cross-fitting.  First, we estimate the control 
     function using a super learner.  In the following step,
      we use the estimated control function to control for 
      endogeneity in the structural equation.  Cross-fitting is done across the individual dimension.  
      The estimator is consistent and asymptotically normal achieving a parametric rate of convergence.       
        We show that
        the Orthogonal SLCF estimator is equivalent to the 
        plug-in IV estimator
        but not equivalent to the
        naive plug-in 2SLS estimator which is not $\sqrt{N_T}$-consistent even 
        with cross-fitting when the first-stage estimator converges more slowly than 
        the parametric rate. 
\end{abstract}


\section{Introduction}

A triangular structural panel data model with additive separable 
individual-specific effects is used to model the causal effect of 
a covariate on an outcome variable when there are unobservable
 confounders, and some of them are time-invariant.  In this setup,
  it is common to assume that the structural equation is linear and to use 
  a linear projection in the reduced form equation.  However,
   a linear specification for the reduced-form equation might be problematic 
   when the 
   conditional mean of the endogenous covariate and the instrumental
    variables is nonlinear in the population.  The reason is that 
    ignoring the nonlinearity 
    could lead to problems due to weak instruments since the strength of the 
    linear relationship might not be high. This paper proposes to exploit the
     nonlinearities in the relationship between the endogenous covariate and 
     the instrumental variables as a way to avoid problems of weak instruments
      due to nonlinearities\footnote{Instruments are weak when they are
       weakly correlated to the endogenous covariate.  
       In a cross-sectional setup, weak IVs cause that
        the median of the 2SLS estimator is biased towards the OLS estimator, 
        2SLS t-test presents size inflation \citep{KeaneARE2024}. In panel data,
         Within-2SLS estimator is consistent when the time dimension grows
          \citep{CaiER2012}.}.
          It is important to note that we focus on studying the issue of weak 
           instruments due to misspecification of the reduced-form equation as 
           well as on providing solutions to this problem.  Importantly, we do not 
           solve the problem of weak-identification since strong instruments are 
still required even after allowing for a flexible specification of the 
 reduced-form equation.  %

An example is the relationship between air pollution and children's 
educational outcomes, where both variables are influenced by 
unobserved confounders that can be either time-varying or time-invariant.
 For instance, neighborhood characteristics and family income are often 
 unobserved in publicly available data. While neighborhood characteristics 
 typically remain constant over time, family income fluctuates. 
 Both factors are relevant to a child's academic achievement 
 and exposure to air pollution, leading to endogeneity in pollution exposure.
  To address this, fixed effects can be used to control for unobserved neighborhood
   characteristics, while an instrumental variable can help account for time-varying 
   confounders. A commonly used instrumental variable is wind direction, 
   which has a well-established nonlinear relationship with air pollution levels
    \citep{ZabrockiN2022}.

This paper studies a triangular structural panel data model in which the reduced 
form equation is nonlinear and individual specific effects are present. It proposes
 the use of Super Learning \footnote{It is feasible to use other machine learning 
 methods instead of Super Learning. However, Super Learning automatically combines
 different methods in an optimal way achieving a risk asymptotically equivalent to that of the best candidate learner.} 
 for the estimation of the nonlinearities in the model. 
  For simplicity consider the following triangular model without covariates for 
  $i \in \{1,...,N\}$,  $t \in \{1,...,T\}$: 
 
  \begin{equation}\label{structural_equation_toy}
y_{it}=x_{1it}\beta_{1_o}+\alpha_{i,y}+\varepsilon_{it},  \quad with \quad
\mathbb{E}[\varepsilon_{it}| x_{1i} ,z_{i}, \alpha_{i,y}]=\mathbb{E}[\varepsilon_{it}| u_{it}],
\end{equation}
\begin{equation}\label{reducedform_equation_toy}
x_{1it}=g_o(z_{it})+\alpha_{i,1x}+u_{it},   \quad with \quad    \mathbb{E}[u_{it}|z_{i},\alpha_{i,1x}]=0, 
\end{equation}
\begin{equation}\label{relationship_errors_toy}
\varepsilon_{it}=\rho u_{it} + \omega_{it}, \quad with \quad 
\mathbb{E}[\omega_{it} |u_{i} ]=0,
\end{equation}

 \noindent where $y_{it}$, and $x_{it}$ are continuous and bounded, $u_i' = [u_{i1}, u_{i2}, ..., u_{iT}]$, $z_i' = [z_{i1}, z_{i2}, ..., z_{iT}]$ are vectors stacking up all the observations available per individual. $\alpha_{i,y}$, and $\alpha_{i,1x}$ represent individual specific effects. \footnote{$\alpha_{i,y}$, and $\alpha_{i,1x}$ can be equal as in the cited example but for generality we allow them to be different.}
 
 The model comprises a structural equation 
 \ref{structural_equation_toy} with an endogenous
  regressor $x_{1it}$ that presents a non-linear
   relationship with the available instrumental
    variable $z_{it}$ \ref{reducedform_equation_toy}.  Both equations present individual effects that represent time-invariant unobserved confounders.  In this setup, conditioning on the individual effects allows us to control for time-invariant unobserved confounders.  But, estimating all the individual effects leads to incidental parameter bias as the number of parameters to estimate grows with the sample size. Then, in order to eliminate the individual effects,  we propose to transform  the model (e.g. first-differencing, or a within transformation which is equivalent to obtain the deviations of the original variables from their individual mean).  On the other hand, the presence of time-varying unobserved confounders is modeled through a linear relationship between the error terms in both equations (\ref{relationship_errors_toy}).  We deal with this issue by exploiting available exclusion restrictions (instrumental variables) along with an Orthogonal Control Function approach on the transformed model.

We derive a Neyman orthogonal score related to the augmented model, and we show that the estimator is equivalent to the IV estimator with instruments equal to 
transformed non-linear relationship between the instruments and the covariates. To estimate the causal parameter of interest $\beta_{1_o}$,
 we provide a two-step estimator.  
  In the first step,  we use a super learner \citep{VanSAMB2007} to 
  learn the nonlinear relationship of the transformed endogenous
   variable with the instrumental variables and the observable confounders.   
   In the second step,  we estimate the parameters of the model solving the derived Neyman orthogonal score. We show that this is equivalent to IV estimation for the $\beta_o$ coefficients that uses $\widehat{\tau g_o}(\tilde{X}_i, z_i)$ as instrument. We call this estimator Orthogonal Super Learner Control 
   Function (SLCF) estimator.  The main results of the paper are that the Orthogonal Super Learner Control Function 
   estimator of $\beta_o$ is consistent and asymptotically normal with a convergence rate equal 
   to $\sqrt{N_T}$ \footnote{$N_T = \sum_i^N T_i$.}. An extensive Monte Carlo experiment is performed 
   to test the small sample properties of the proposed estimator. 
   We conclude that the estimator performs well provided that we can accurately
    learn the nuisance parameter in the first stage.     
An R package is available from the author. 

Our paper contributes to the panel data literature in several ways. 
First, this paper studies the use of machine learning to estimate 
panel data models with an endogenous covariate (arising from unobserved confounders 
 or measurement error) and additive individual fixed effects using a control 
 function approach and exploiting nonlinearities. Second, it provides general theoretical results that apply to different
  panel data transformations \footnote{Panel data transformations
  are required to eliminate the individual  specific effects} while using 
  machine learning methods and control function. Third, it shows that, under suitable 
  conditions, one can identify the transformed nuisance parameter and as a result 
  the transformed control function. Fourth, it provides a Neyman orthogonal score.
    Finally, it shows that the orthogonal control function estimator is equivalent to the
  plug-in IV estimation. 

  In contrast, the available papers using machine learning to estimate panel data
   models focus on exogenous covariates (time-varying confounders are observed) and 
   the presence of individual fixed effects and focus on specific panel data transformations.
   For instance, \cite{klosinWP2022} propose a 
   debiased LASSO estimator to estimate heterogeneous treatment effects of a 
   continuous treatment, which is modeled using a structural nonparametric panel
    data model with additive fixed effects and strictly exogenous covariates. They
     account for the individual specific effects by taking the first-difference of
      the original model.  \cite{ SemenovaQE2023} focus on a structural equation
       model for panel data allowing for heterogeneous treatment effects with 
       unconfoundedness on observables (sequentially exogenous covariates) and individual specific fixed effects. %
In order to deal with the individual fixed effects, the authors use a Mundlak approach.  They provide an algorithm composed of three steps.
 In the first step, using cross-fitting they estimate the nuisance functions of the structural and treatment equations to obtain the residuals. In the second step, they estimate the conditional average treatment effect using lasso estimators.   In the third step,  they provide debiased inference.

This paper also contributes to the literature on weak instruments in panel
 data by being the first to focus on a setup with \textit{short time dimension} in
 the presence of nonlinearities. In contrast, existing papers on
  weak instrumental variables in panel data focus on the asymptotic properties of 
  the Within 2SLS estimator providing conditions for consistency
  under weak instruments 
  when the \textit{time dimension grows to infinity}. \cite{CaiER2012} show that the Within 2SLS is consistent when the 
 individual and time dimension grow to infinity. They also present asymptotically 
 pivotal tests (Anderson-Rubin test, Kleibergen test) for longitudinal data.

 In Section
\ref{S_model_line_nonlin}, we present the baseline model. In Section 
\ref{S_Identification}, we describe the identification strategy, and we derive a 
Neyman orthogonal score. In Section 
\ref{S_SLCF}, we present the SLCF estimator. In Section \ref{S_Properties}, we 
  present the large sample properties of the SLCF estimator. 
  In Section \ref{S_Comparison}, we present a comparison of the proposed estimator
  with IV and 2SLS plug-in estimators. In Section \ref{S_Simulation}, we present
   the simulation experiment and its results, 
in Section \ref{S_EA} we present the results of
   the empirical application, and in 
   Section \ref{S_Conclusions} we present the
    conclusions of the paper.

\section{Other Related literature}

This paper is related to the literature on semiparametric panel data models.   
\citet{RodriguezPoo2017JES} explain that one can either use a first-differencing 
approach or a profiling technique to estimate a semiparametric model with an 
additive separable individual specific effect and a disturbance term. 
\citet{Li1996EL} explains that kernel estimation of 
the conditional expectation of the first-differenced endogenous 
regressor is not feasible when the dimension of the conditioning 
set is larger than 5.  As a solution,  \citet{baltagi2002series}
 presents an additive estimation of the first-differenced 
 non-linear function.  

In addition, this paper is related to the literature on orthogonal/debiased 
estimation, which focus on providing  $\sqrt{N}$ consistent estimators of
 low-dimensional parameters in the presence of high-dimensional nuisance 
 parameters.  \citet{Chernozhukov2018EJ} define an orthogonal score as the 
 one that presents a vanishing Gateaux derivative with respect to the nuisance
  parameters when evaluated at the true finite-dimensional parameter values. 
    The orthogonal score is closely related to the score proposed by 
    \citet{Robinson1988E}.  

Finally, our work relates to the literature of mixed effects models and machine 
learning.  Our model is related to a mixed effects model assuming that covariates 
are exogenous conditioning on the cluster random effects. \cite{EmmeneggerSJS2023}
 use double machine learning to estimate a partially linear mixed-effects model.  
 Our setup differs from \cite{EmmeneggerSJS2023} as it allows for endogeneity, 
 it includes a control function leading to the study of a different score function
 than the one presented by the authors.
  \cite{Young2024JRSSB} propose a sandwich boosting for accurately estimating partially linear models for grouped data. 

\subsection{Super Learner}

The Super Learner was proposed by \cite{VanSAMB2007}, it is a heterogeneous ensemble method within 
Supervised Learning. It consists of different base learners that are trained on the training data set, then
their predictions are combined using a meta-learner by minimizing the cross-validated risk. The meta-learner
can be a linear regression of the predictions of each learner on the outcome variable.  
\cite{VanSAMB2007} show that
the Super Learner performs asymptotically as well as or better than any learner in the library of base learners.

\section{ The Setup}\label{S_model_line_nonlin}

We study a model composed of a linear structural equation, the outcome variable ($y_{it}$) explained by an endogenous variable ($x_{1it}$), $K-1$ exogenous regressors ($\tilde{x}_{it}$), and an additive individual-specific effect ($\alpha_{i,y}$). The model presents an external instrumental variable ($z_{it}$) for the endogenous regressor. The endogenous regressor is mean dependent on the exogenous regressors,  the instrumental variable, and the individual specific effects $\alpha_{i,1x}$. We assume that: 
\sloppy
\begin{assumption}\label{A_data}
We observe a sequence of data sets $\{O_i = (y_{i},x_{1i},\ldots,x_{Ki},z_{i})\}_{i=1}^N$,
 which are $T_i\times 1$ independent copies of the $K+2$ random vectors
  $(\mathbf{y},\mathbf{x}_1,\mathbf{x}_2, \ldots,\mathbf{x}_K,\mathbf{z})$ 
  taking values in $[-M_y,M_y]^T \times [-M_{x1},M_{x1}]^T \times [-M_{xK},M_{xK}]^T \times [-M_z,M_z]^T$ 
  with $M_y<\infty$, $M_{xk}<\infty$,  $M_z<\infty$. Thus, each element $t$ of
   $\mathbf{y}$ is bounded and contained in $[-M_y,M_y]$.  Similarly, each element of
    $x_k$ is contained in $[-M_{xk},M_{xk}]$, which means that $\mathbf{x}_{kt}$ 
    is a continuously distributed random variable with compact support $[-M_{xk},M_{xk}]$.
     In addition, each element of $\mathbf{z}$ is a continuously distributed random variable with compact support $[-M_z,M_z]$. 
     Finally, $T_i$ is uniformly bounded such that $T_i < T_{max}$.
\end{assumption}
Assumption \ref{A_data} states that each individual $i$ presents $T_i$ observations.
  The total number of observations is denoted by $N_T = \sum_i T_i$.  %

\begin{assumption}\label{A_structural_eq}
 The structural equation has a linear form with an additive unobserved error term and an additive individual specific effect $\alpha_{i,y}$
 
 \begin{equation}
  \begin{gathered}
y_{it}=x_{1it}\beta_{1_o}+\tilde{x}_{it}'\beta_{2_o}+\alpha_{i,y}+\varepsilon_{it}, \quad i \in \{1,...,N\},  t \in \{1,...,T_i\}, \\
\mathbb{E}[\varepsilon_{it}|\tilde{x}_{i}, z_{i}, \alpha_{i,y}]=
\mathbb{E}[\varepsilon_{it}|\tilde{x}_{i}, \alpha_{i,y}]=0.     
 \end{gathered}
 \end{equation}

\noindent where $\tilde{x}_{it}=[x_{2it}, x_{3it}, ..., x_{Kit}]'$, 
$z_i = [z_{i1}, \cdots, z_{iT}]'$.  

\noindent  The unobserved random term $\varepsilon_{it}$ is independent and identically
 distributed with a continuous cumulative distribution function. It has a compact 
 support with zero mean conditional on  $\tilde{x}_{i}=[\tilde{x}_{i1}, \tilde{x}_{i2}, \ldots,
  \tilde{x}_{iT_i}]'$, $z_{i}=[z_{i1}, z_{i2},..., z_{iT_i}]'$, $\alpha_{i,y}$, and 
   $Var(\varepsilon_{it}|\tilde{x}_{i}, z_{i}, \alpha_{i,y})< \infty$. $\alpha_{i,y}$ 
   are individual specific
   effects with sigma-field $\mathcal{A}_y$.
    The unobserved parameter
    of interest is $\beta_{1_o}$.
\end{assumption}

\noindent Under assumption \ref{A_structural_eq}, $\tilde{x}_{it}$ are strictly exogenous in
 the structural equation conditional on $\alpha_{i,y}$.  In addition, $\{z_{is}, \quad s=1,\dots,T \}$
   are uncorrelated with the structural error term $\varepsilon_{it}$ conditional on $\alpha_{i,y}$. 

\begin{assumption}\label{A_reduced_form}
The regressor $x_{1it}$ is endogenous and has an unknown nonlinear reduced form equation with an 
additive separable individual specific effect $\alpha_{i,1x}$: 
\begin{equation}
\begin{gathered}    
x_{1it}=g_o(\tilde{x}_{it},z_{it})+\alpha_{i,1x}+u_{it}, \quad i \in \{1,...,N\},  t \in \{1,...,T_i\}, \\
\mathbb{E}[u_{it}|\tilde{x}_{i},z_{i},\alpha_{i,1x}]=0,\\
\end{gathered}
\end{equation}

 \noindent with the unknown function $g_o: [-M_x,M_x]^K \times [-M_z,M_z]\rightarrow \mathcal{S} \subset \mathbb{R}$, with 
 $g_o \in \mathcal{H}(\varsigma,C)$ where $\mathcal{H}(\varsigma, C)$ represents a H\"older-class with smoothness parameter
 equal to $\varsigma \geq 2$\footnote{If the unknown function is linear, the SLCF estimator is still consistent
 if instruments are not weak.  We argue that if the unknown function $g_o$ is linear, 
 the library of the base learners includes the linear regression, and the
 instruments are appropriately transformed, then the SLCF estimator is equivalent to the 2SLS estimator. 
 A proof of this statement is out of the scope of this paper.}. $\alpha_{i,1x}$ are individual specific
   effects measurable with respect to $\mathcal{A}_y$, therefore knowledge of $\alpha_{i,y}$ implies 
   knowledge of $\alpha_{i,1x}$. The reduced form error term $u_{it}$ is independent and identically distributed 
 with a continuous cumulative distribution function. It has a compact support 
 with zero mean conditional on  $\tilde{x}_{i}=[\tilde{x}_{i1},
  \tilde{x}_{i2}, \ldots, \tilde{x}_{iT_i}]'$, 
  $z_{i}=[z_{i1}, z_{i2},..., z_{iT_i}]'$, and  $Var(u_{it}|\tilde{x}_{i}, z_{i}, \alpha_{i,1x})< \infty$. 
 
\end{assumption}

 \noindent Under assumption \ref{A_reduced_form}, $\tilde{x}_{it}$,  $z_{it}$ are strictly exogenous in the reduced form equation conditional on the individual specific effect $\alpha_{i,1x}$.  In addition, this assumption states that the endogenous regressor has a nonlinear relationship with the exogenous regressors $\tilde{x}_{it}$, and the variable $z_{it}$.

As a consequence of Assumptions \ref{A_structural_eq} and \ref{A_reduced_form}, $z_{it}$ is an instrumental variable.  

\noindent The decomposition of $x_{1it}$ into a conditional expectation and an additive disturbance term is appropriate under the assumption that $x_{1it}$ is continuous (Assumption \ref{A_data}).     Assumption \ref{A_reduced_form} needs to be modified when $x_{1it}$ is a non-continuous endogenous explanatory variable.  

\begin{assumption} \label{A_eps_u_omega}
The structural error term $\varepsilon_{it}$  has a linear relationship with $u_{it}$:
\begin{equation}
\begin{gathered}
\varepsilon_{it}=\rho u_{it} + \omega_{it} , \quad i \in \{1,...,N\},  t \in \{1,...,T_i\},   \\
\mathbb{E}[ \omega_{it} |\tilde{x}_{i}, z_i, \alpha_{i,y}, u_{i}] = \mathbb{E}[ \omega_{it} |u_{i}]=0.
\end{gathered}
\end{equation}

\noindent  The error term $\omega_{it}$ is independent and identically 
distributed with a continuous cumulative distribution. It has a compact 
support with zero mean conditional on $u_{i}=[u_{i1}, u_{i2}, \dots, u_{iT_i}]'$,
 and  $Var(\omega_{it}|u_{it}) = \sigma^2_{\omega}< \infty$.
\end{assumption}
As a consequence of Assumption \ref{A_eps_u_omega}, $\omega_{it}$ is uncorrelated 
with $u_{it}$. 

\begin{remark}
    The assumption of a linear relationship between $\varepsilon_{it}$ and $u_{it}$ is a 
    crucial assumption in panel data setups.  
    Relaxing this assumption to allow for a nonlinear relationship with $u_{it}$ is not trivial even 
    when there are only additive individual fixed effects. Since one needs to eliminate fixed effects
    through transforming the data, one can only identify the transformed reduced form errors
    under assumptions \ref{A_structural_eq}, and \ref{A_reduced_form}. 
     However, identification under a nonlinear function of $u_{it}$ requires knowledge of the 
     untransformed reduced form errors which is not feasible (See Section \ref{S_Identification}). 
     For example, in the case of first-differencing, we only identify the 
    first-differenced reduced-form errors but we would need the current and first-lag of
     the reduced form errors in the required information set. These reduced-form errors 
    are not separately identifiable under assumptions \ref{A_structural_eq}, and \ref{A_reduced_form}.   
\end{remark}

As consequence of Assumptions \ref{A_reduced_form} and \ref{A_eps_u_omega}, 
the structural error term $\varepsilon_{it}$ is correlated with $x_{1it}$. The reason 
is that $\mathbb{E}[x_{1it}\varepsilon_{it}]=\mathbb{E}[(g_o(\tilde{x}_{it},z_{it})
+\alpha_{i,1x}+u_{it})\varepsilon_{it}]$ by Assumption \ref{A_reduced_form}.  
The last expression is equal to $\mathbb{E}[(g_o(\tilde{x}_{it},z_{it})+
\alpha_{i,1x}+u_{it})(\rho_o u_{it}+\omega_{it})]=\rho_o\mathbb{E}[u_{it}^2]$ 
by Assumptions \ref{A_reduced_form} and \ref{A_eps_u_omega}.

\noindent The error term $\omega_{it}$ has zero mean conditional on 
$x_{1it}$,  $\tilde{x}_{it}$, $\alpha_{i,y}$, and $u_{it}$
 ($\mathbb{E}[\omega_{it}|x_{1it},\tilde{x}_{it},u_{it},\alpha_{i,y}]=0$) 
 because $\mathbb{E}[ \omega_{it} |x_{1i}, \tilde{x}_{i}, z_i, \alpha_{i,y}]
 =\mathbb{E}[\omega_{it}|u_{i}, \tilde{x}_{i},z_i,\alpha_{i,y}]=\mathbb{E}[\omega_{it}\mid u_i]=0$.  The first equality 
is a consequence of the fact that $x_{1it}$ is a one-to-one function
 with $u_{it}$ conditional on $\tilde{x}_{it}$, $z_{it}$, $\alpha_{i,1x} 
 \subseteq \alpha_{i,y}$  \citep{Wooldridge2010Book}. The second equality holds by
 Assumption \ref{A_eps_u_omega}.  
 Notice that a sufficient condition for mean independence of the instrumental variable with the structural and reduced form errors is
  that  it is independent of both errors.

\subsection{Relationship with the Potential Outcomes Framework}\label{Potential}

 In this subsection, we link our setup with the potential outcomes framework.  Following \cite{holland1988causal} and 
 \cite{guo2016control}, we define the potential
  outcome that would be observed if the treatment is continuous $x_{1,it} = 
  x_{1,it}^*$ as $y_{it}(x_{1,it}^*)$. We
 assume an infinite population setup over the cross-section. 

\begin{assumption}\label{Assumption:IH}
     Invariance to history

  \begin{equation}
  y_{it}(x_{1,i}^*) = y_{it}(x_{1,it}^*),
  \end{equation}

  with $x_{1,i} = [x_{1,i1}, ..., x_{1,iT}] $.
 \end{assumption}

 \begin{assumption}\label{Assumption:ITR}
    Individualistic treatment response (ITR)
 \begin{equation}
    y_{it}(x_{1,t}^*) = y_{it}(x_{1,it}^*),
 \end{equation}
   
 with $x_{1,t} = [x_{1,1t}, \ldots, x_{1,Nt}] $.
 \end{assumption}
Assumptions \ref{Assumption:IH} and \ref{Assumption:ITR} state that the potential outcome of 
individual $i$ at period $t$ depends only on the current level of treatment 
and on own treatment level. 
 
\begin{assumption}\label{Assumption:Consistency}
    $y_{it} = y_{it}(x_{1it}) \quad  \text{when} \quad x_{1it} = x_{1it}^*$ 
 \end{assumption}

 \begin{assumption}\label{Assumption:potential_outcome}
 The potential outcome $y_{it}(x_{it}^*)$ has the following linear additive form:    
 \begin{equation}
     y_{it}(x_{it}^*) = y_{it}(0) + x_{1it}^* \beta_{1_o}.  
 \end{equation}
 \end{assumption}
 This assumption states that the potential outcome of individual $i$ at period $t$ if treatment 
$x_{1,it}$ is equal to $x_{1,it}^*$ has a linear relationship with the 
potential outcome without treatment and the treatment.  In addition, we assume 
homogeneous treatment effects $\beta_{1o}$. The latter is a restrictive assumption, and 
relaxing it to heterogeneous treatment effects is left for future research.
Notice that the instrumental variable does not have an effect on the potential outcome. 

 \begin{assumption}\label{Assumption:notreatment_potential_outcome}
 The potential outcomes without any treatment have a linear conditional mean 
 on the pre-treatment variables $\tilde{x}_{it}$, and individual fixed effects:

 \begin{equation}
     \mathbb{E}[y_{it}(0)|\tilde{x}_{it}, \alpha_{i,y}] = \tilde{x}_{it}' \beta_{2_o} + \alpha_{i,y}.
 \end{equation}     
 \end{assumption}
This assumption states that the potential outcome of individual $i$ at period $t$ without treatment
  presents a linear relationship with the 
covariates $\tilde{x}_{it}$ and the individual fixed effects.  Importantly, the treatment does not
have an effect on the covariates. This assumption can be relaxed to allow for 
a nonparametric function of the pre-treatment covariates. In this case the estimation procedure proposed is still
valid after obtaining the deviations from the conditional expectation on the pre-treatment
covariates (\citealt{Robinson1988E}, \citealt{Chernozhukov2018EJ}).

 \begin{assumption}\label{Assumption:instruments}
       \begin{align}  
        & \{x_{1it}(z_{it}^*), y_{it}(x_{1it}^*)\} \perp z_{i} \mid \tilde{x}_{i}, \alpha_{i,1x} \label{Assumption:instruments1} \\
    & x_{1it} = x_{1it}(z_{it}) \quad  \text{when} \quad z_{it} = z_{it}^*  \label{Assumption:instruments2}   \\
    & \mathbb{E}[x_{1it}(z_{it})|\tilde{x}_{it}, \alpha_{i,1x}] = g_o(\tilde{x}_{it}, z_{it}) + \alpha_{i,1x} \label{Assumption:instruments3}
    \end{align}
\end{assumption}
This assumption states that the instrument is independent of the outcome 
potential outcomes, and the continuous treatment potential outcomes conditional
on the covariates and the individual specific effects. Condition
\ref{Assumption:instruments1} of assumption \ref{Assumption:instruments} is 
stronger than assumption \ref{A_reduced_form} as 
it requires conditional independence of the instrument.

 Assumptions \ref{Assumption:IH} to \ref{Assumption:notreatment_potential_outcome}
 , and condition \ref{Assumption:instruments1} of assumption \ref{Assumption:instruments}
 imply Assumption \ref{A_structural_eq}.
Assumption \ref{Assumption:instruments} implies assumption \ref{A_reduced_form}.
 \subsection{Estimand}
The causal estimand of interest is the average treatment effect: 
 \begin{equation}
     ATE = \mathbb{E}[y_{it}(x_{it}^*+1)] - \mathbb{E}[y_{it}(x_{it}^*)] = \beta_{1_o}.
 \end{equation}

\section{Identification} \label{S_Identification}

The identification of the parameter of interest $\beta_{1_o}$ relies on the 
presence of an instrumental variable $z_{it}$ conditioning on the 
individual-specific effects (Assumptions \ref{A_structural_eq} and 
\ref{A_reduced_form}).   Since we run into an incidental parameter problem,
 we get rid of the individual-specific effects by using an appropriate 
 transformation of the data $\tau$. The typical transformations of the data are 
 the first-differencing and within transformations.  The first-difference
  transformation is equal to $y_{it}-y_{it-1}$, and the Within transformation
   is equal to $y_{it}-\sum_{t=1}^{T_i} y_{it}/T_i$. 

\begin{equation} \label{dif_structural}
\tau y_{it}=\tau x_{1it}\beta_{1_o}+\sum_{k=2}^K \tau x_{kit}\beta_{2k_o}+\tau\varepsilon_{it}, i\in \{1, ..., N\}, t \in \{t_a, ..., T_i\},
\end{equation}

\begin{equation} \label{dif_reduced}
\tau x_{1it}=\tau g_o(\tilde{x}_{it}, z_{it})+\tau u_{it}, i\in \{1, ..., N\}, t \in \{t_a, ..., T_i\}.
\end{equation}
\noindent where $t_a$ is the first time period available after transforming the data.  In the case of first-differencing $t_a=2$, and for a Within-transformation $t_a=1$.

 Now,  it is tempting to substitute equation \ref{dif_reduced} into \ref{dif_structural} to obtain: 
\begin{equation}\label{naive_model}
\tau y_{it}=\tau g_o(\tilde{x}_{it},z_{it})\beta_{1_o}+\sum_{k=2}^K \tau x_{kit}\beta_{2k_o}+\tau u_{it}\beta_{1_o}+\tau\varepsilon_{it}, i\in \{1, ..., N\}, t \in \{t_a, ..., T_i\}.
\end{equation}
\noindent In equation \ref{naive_model},  one can believe that the estimation of the
 unknown parameters is straightforward since the error term 
 $\tau u_{it}\beta_{1_o}+\tau \varepsilon_{it}$ in \ref{naive_model} has zero mean
  conditional, on an appropriate information set, as a consequence of assumptions
   \ref{A_structural_eq}, and \ref{A_reduced_form}.   Moreover, one can think that
    the best estimation procedure is to estimate $\tau g_o(\tilde{x}_{it},z_{it})$ 
    in the first stage and in the second stage, plug it in equation 
    \ref{naive_model} and perform simple OLS.  But this is not possible because the
     regularization bias of the first stage estimator of 
     $\tau g_o(\tilde{x}_{it},z_{it})$ contaminates the estimation of $\beta_{1_o}$ 
     and $\beta_{2_o}$ \citep{Chernozhukov2018EJ, Robinson1988E, guo2026robustness} and 
     produces estimates that are not $\sqrt{N_T}$-consistent (See Section \ref{S_Comparison} for more details). %

In order to deal with this problem,  we can use a control function approach and augment the structural equation by controlling for the unobservable $\tau u_{it}$  as follows: 
\begin{equation} \label{ControlFunction_eq}
\tau y_{it}=\tau  x_{1it}\beta_{1_o}+\sum_{k=2}^K \tau x_{kit}\beta_{2k_o}+\rho_o \tau  u_{it} + \tau \omega_{it}, i\in \{1, ..., N\}, t \in \{t_a, ..., T_i\}.
\end{equation}
In model \ref{ControlFunction_eq},  under assumptions \ref{A_structural_eq} - \ref{A_eps_u_omega} we can guarantee zero correlation between 
$\tau x_{1it}$, and $\tau  \omega_{it}$ thanks to the presence of an instrumental variable in the reduced form equation that is mean independent 
of the error term $\omega_{it}$.

Stacking up the time observations and letting $M_{\tau_i}$ denote
the transformation operator mapping 
$\mathbb{R}^{T_i} \rightarrow \mathbb{R}^{T_{i,a}}$, with $T_{i,a}$ equal
 to $T_i$ under the within transformation and $T_i-1$ under the
  first-difference transformation, we obtain the following moment conditions: 
\begin{equation}\label{non_orthogonal_score}
\mathbb{E}[\phi(O_i;\theta_o,\tau g_o)]=\textbf{0},
\end{equation}
\noindent where 
$\phi(O_i;\theta_o,\tau g_o)$ is a vector of score functions equal to
 $(M_{\tau_i} H_{i})' V_i^{-1} M_{\tau_i} \omega_{i}  $,  $\theta_o=[\beta_{1_o} 
 \quad \beta_{2_o}' \quad  \rho_o]'$, $M_{\tau_i} \omega_{i}=M_{\tau_i} 
 y_{i}-M_{\tau_i} x_{1i}\beta_{1_o}-\sum_{k=2}^K M_{\tau_i} x_{ki}\beta_{2k_o}
 -\rho_o M_{\tau_i} u_{i}$,  $\textbf{0}=[0\quad \cdots \quad 0\quad 0]'$, 
  and $M_{\tau_i} H_{i}=[M_{\tau_i} x_{1i}\quad M_{\tau_i}\tilde{X}_{i}\quad 
  M_{\tau_i} u_{i}]$, $H_{i}=[x_{1i}\quad \tilde{X}_{i}\quad 
  u_{i}]$, and $V_i$ the weighting matrix corresponding to $M_{\tau_i}M_{\tau_i}'$
   under the first-difference transformation since  $\Sigma_i= \mathbb{E}[(M_{\tau_i}\omega_i)(M_{\tau_i}
  \omega_i)'| H_i, z_i]=\sigma^2_{\omega}M_{\tau_i}M_{\tau_i}'=\sigma^2_{\omega} 
  V_i $,
   and equal to the identity matrix $I_i$ under the within transformation.

\begin{remark}
If we use the first-difference transformation, $M_{\tau_i}=D_i=\begin{bmatrix} 
    1 & -1 & 0& \cdots & 0 \\
    0 &1 & -1 &  \cdots &  0  \\
    \vdots &   & \ddots & \ddots & \vdots  \\ 
    
    0 & \cdots &  0 & 1 & -1  \\ 
\end{bmatrix}$, with dimensions $(T_i-1) \times T_i  $.  If we use the within transformation, $M_{\tau_i}= W_i= I_i-1/T_i\begin{bmatrix} 
    1 &   1 &   1& \cdots & 1 \\
    1 &1 & 1 &  \cdots &  1  \\
    \vdots &  \vdots  & \ddots & \ddots & \vdots  \\ 
    
    1 & \cdots & 1 & 1 & 1  \\
\end{bmatrix}$ , with dimensions $T_i \times T_i  $, $I_i$ is the identity matrix.
\end{remark}
\begin{remark}
   The moment condition is valid even under misspecification of $V_i$ because we assume strict exogeneity.  If we would like to relax the assumption of strict exogeneity to weak exogeneity, then it is better to use $\mathbb{E}[(M_{\tau_i}H_i)'M_{\tau_i}\omega_i]=0$ as moment conditions. 
   \end{remark}
   \begin{remark}
   If $M_{\tau_i}=W_i$ (the within transformation), the moment condition
    can be simply written as $\mathbb{E}[(M_{\tau_i}H_i)'M_{\tau_i}\omega_i]=0$.   
\end{remark}

Finally, we identify $\tau u_{it}$ since it has zero mean conditional on an
 appropriate information set $A_i=\{\tilde x_{it-j}, z_{it-j},  L_i \leq j \leq U_i \}$
  with $L_i \geq t_i - T_i$ and $U_i \leq t_i-1$ as a consequence of Assumption
   \ref{A_reduced_form}. This identification strategy was first proposed 
   by \cite{AvilaMarquez2023}, and it was presented at the 28th International Panel Data Conference at the University of Amsterdam in 2023.

  \begin{proposition}[Identification $\theta_o$] \label{Proposition_Identification}
     If $\{O_i = (y_i, X_i, z_i) \}_{i\in [N]}$ satisfy assumptions
    \ref{A_data} to \ref{A_eps_u_omega}, $\mathbb{E}[\phi(O_i;\theta_o,\tau g_o)]$
     is continuously differentiable 
    at $\theta_o$, and $\mathbb{E}[(M_{\tau_i} H_{i})'
    V_i^{-1}M_{\tau_i} H_{i}]$ has full rank, 
    then $\theta_o$ is locally identified,
    i.e., $\exists \zeta>0$, such that $\mathbb{E}[\phi(O_i;\theta,\tau g_o)] \neq \textbf{0}$ for 
    $\theta \neq \theta_o$, $\theta \in {\Theta_o}$ with $\Theta_o = \{\theta: ||\theta - \theta_o||_P < \zeta \}$ .\\ 
  \end{proposition}

  \noindent Proposition \ref{Proposition_Identification} states that the parameters of 
interest are locally identified if the vector of moment conditions is differentiable at 
the true parameter value, and the matrix $\mathbb{E}[(M_{\tau_i} H_{i})'V_i^{-1}M_{\tau_i} 
H_{i}]$ is full rank. 

  Proposition \ref{Proposition_Conditions} provides conditions under which this 
rank condition holds.

  \begin{proposition}[Conditions]\label{Proposition_Conditions}
     If $\{O_i = (y_i, X_i, z_i) \}_{i\in [N]}$ satisfy assumptions
    \ref{A_data} to \ref{A_eps_u_omega}, $M_{\tau_i} x_{1i}$, and $M_{\tau_i}\tilde{X}_{i}$ 
    are not perfectly multicollinear, then
    
    i) $\mathbb{E}[(M_{\tau_i} H_{i})'V_i^{-1}M_{\tau_i} H_{i}]$ has a minimum eigenvalue 
    $\lambda_{min} >0$ if and only if $M_{\tau_i} g_o(\tilde{x}_{i}, z_{i}) \notin span(M_{\tau_i}\tilde{X}_{i})$. \\ 
    
    ii) $\mathbb{E}[(M_{\tau_i} H_{i})'V_i^{-1}M_{\tau_i} H_{i}]$ has a minimum eigenvalue 
    $\lambda_{min} = 0$ if and only if $M_{\tau_i} g_o(\tilde{x}_{i}, z_{i}) \in span(M_{\tau_i}\tilde{X}_{i})$. 
  \end{proposition}
  The proof is straightforward and omitted. 

  This proposition states that the matrix $\mathbb{E}[(M_{\tau_i} H_{i})'V_i^{-1}M_{\tau_i} H_{i}]$
   is ill-conditioned if $\tau g_o(\tilde{x}_{i}, z_{i})$ lies in the column span of $M_{\tau_i}\tilde{X}_{i}$.
   This could happen if the variation of the transformed nuisance parameter does not provide
  any different variation from the transformed covariates.  This includes the degenerate situation 
   when the transformed nuisance is equal to the null vector which intuitively means that
   $\tau u_{it}$ approaches $\tau x_{1it}$ causing high multicollinearity.
    Alternatively, if the functional form of $g_o(.)$ is linear, ill-conditioning 
  happens if instruments are weakly correlated with the endogenous covariates. 
  In those situations, the 
   instruments do not provide additional information permitting the identification of the 
   parameter of interest. Thus,
  the Gram matrix becomes ill-conditioned as instruments are weak under the conditions
  cited in the proposition.

  \begin{remark}
    Notice that if $z_{it}$ is not present in $g_o(.)$, we can still achieve 
    identification if $\tau g_o(\tilde{x}_{it})$ does not lie in the column span of $\tau \tilde{x}_{it}$. 
    However, identification would rely on a functional form assumption which 
    is a strong assumption. 
  \end{remark}

\begin{example_un}{ First-difference transformation:}

We could follow \cite{baltagi2002series} by first-differencing the structural 
and reduced-form equations such that we obtain:

\begin{equation} \label{dif_structural_ex}
\Delta y_{it}=\Delta x_{1it}\beta_{1_o}+\Delta\tilde{x}_{it}'\beta_{2_o}+\Delta\varepsilon_{it}, i\in \{1, ..., N\}, t \in \{2, ..., T_i\}.
\end{equation}

\begin{equation} \label{dif_reduced_ex}
\Delta x_{1it}=\Delta g_o(\tilde{x}_{it}, z_{it})+\Delta u_{it}, i\in \{1, ..., N\}, t \in \{2, ..., T_i\}.
\end{equation}
Now,  we can use the proposed control function approach and augment the structural equation by controlling for the unobservable $\Delta u_{it}$  as follows: 
\begin{equation} \label{ControlFunction_eq_ex}
\Delta y_{it}=\Delta  x_{1it}\beta_{1_o}+\Delta  \tilde{x}_{it}'\beta_{2_o}+\rho_o \Delta  u_{it} + \Delta \omega_{it}, i\in \{1, ..., N\}, t \in \{2, ..., T_i\}.
\end{equation}
Next,  we set up the following moment conditions: 
\begin{equation}\label{non_orthogonal_score_ex}
\mathbb{E}[(D_i H_{i})' \Sigma_i^{-1} D_i \omega_{i}]=\textbf{0},
\end{equation}
\noindent where 
  $D_i \omega_{i}=D_i y_{i}-D_i x_{1i}\beta_{1_o}-D_i \tilde{x}_{i}'\beta_{2_o}
  -\rho_o D_i u_{i}$,  $\textbf{0}=[0\quad \cdots \quad 0\quad 0]'$, 
  $D_i H_{i}=[D_i x_{1i}\quad D_i\tilde{x}_{i}'\quad D_i u_{i}]'$, and
   $\Sigma_i = \sigma^2_{\omega} D_iD_i'$.

Finally, we identify $\Delta u_{it}$ since it has zero mean conditional on $\tilde{x}_{it}$, $\tilde{x}_{it-1}$, $z_{it}$, $z_{it-1}$ as a consequence of Assumption \ref{A_reduced_form}.
\begin{proposition}[Identification transformed nuisance parameter]\label{Ident_nuisance}
    Under assumption \ref{A_reduced_form}, $\Delta g({\tilde{x}_{it},z_{it}})$ is identified by $\mathbb{E}[\Delta x_{1it} | \tilde{x}_{it},z_{it},\tilde{x}_{it-1},z_{it-1}]$. 
\end{proposition}
\begin{proof}
\begin{equation}
    \mathbb{E}[\Delta x_{1it}|\tilde{x}_{it},z_{it},\tilde{x}_{it-1},z_{it-1}] =   \mathbb{E}[\Delta g(\tilde{x}_{it},z_{it})|\tilde{x}_{it},z_{it},\tilde{x}_{it-1},z_{it-1}] + \mathbb{E}[\Delta u_{it}|\tilde{x}_{it},z_{it},\tilde{x}_{it-1},z_{it-1}]        
\end{equation}
    By the linearity property of the expectation operator, we have that the first term  is equal to: 
    \begin{equation}
    \begin{split}
        \mathbb{E}[ g(\tilde{x}_{it},z_{it})|\tilde{x}_{it},z_{it},\tilde{x}_{it-1},z_{it-1}] - \mathbb{E}[ g(\tilde{x}_{it-1},z_{it-1})|\tilde{x}_{it},z_{it},\tilde{x}_{it-1},z_{it-1}] = \\ 
        g(\tilde{x}_{it},z_{it})- g(\tilde{x}_{it-1},z_{it-1})
    \end{split}        
    \end{equation}
    Similarly for the second term, by the Law of Iterated Expectations and Assumption \ref{A_reduced_form}: 
     \begin{equation}     
      \mathbb{E}[ u_{it}|\tilde{x}_{it},z_{it},\tilde{x}_{it-1},z_{it-1}] - \mathbb{E}[ u_{it-1}|\tilde{x}_{it},z_{it},\tilde{x}_{it-1},z_{it-1}] =0
     \end{equation}     
\end{proof}

\end{example_un}

\begin{remark}
A similar analysis and conclusion is obtained using a Within transformation 
which is equivalent to time demeaning the data.  In this case, the information 
set $A_i = \{ \tilde{x}_{it-j}, z_{it-j}, -T_i+t \leq j \leq t-1\}$. 
\end{remark}
In Section \ref{S_SLCF}, we propose an estimation method based on the Control Function approach.

\subsection{Weak identification vs weak instruments due to ignored nonlinearities}\label{Section_Weak}

In this subsection, we characterize the problem of weak instruments 
caused by ignored nonlinearities in the reduced-form equation. 
Following \cite{mikusheva2022inference}, \cite{guo2026robustness}, 
and \cite{baltagi2012estimation},  we propose the following generalized concentration 
parameter in the population:

\begin{equation}
\pi_o = \mathbb{E} \Big[(M_{\tau_i} g_o(\tilde{X}_i, z_i)-M_{\tau_i}\tilde{X}_i b_2)'\Phi_{\tau u,i}^{+}
 (M_{\tau_i} g_o(\tilde{X}_i, z_i)-M_{\tau_i}\tilde{X}_i b_2))
\Big],
\end{equation}

\noindent with $b_2 = \big(\tilde{X}'M_{\tau}'\Phi_{\tau u}^{+}M_{\tau}\tilde{X}\big)^{-1}
\big(\tilde{X}'M_{\tau_i}'\Phi_{\tau u}^{+} M_{\tau_i}g_o(\tilde{X}, z)\big)$, 
$\Phi_{\tau u}^+$ a block-diagonal matrix with $\Phi_{\tau u,i}^{+}$ in the diagonal
blocks, $\Phi_{\tau u,i}^{+}$ the Moore-Penrose inverse, 
and $\Phi_{\tau u,i}= \mathbb{E}[M_{\tau_i} u_i 
 u_i' M_{\tau_i}']$. 

If $\pi_o$ is bounded away from 0, we do not have an issue of weak identification.
Under conditions presented in proposition \ref{Proposition_Conditions}, we guarantee 
that $\pi_o$ is bounded away from 0. However, we do not impose a lower bound. 

\begin{proposition}[Population concentration parameter inequality]\label{Proposition_weak_IV}
 If $\{O_i = (y_i, X_i, z_i) \}_{i\in [N]}$ satisfy assumptions \ref{A_data} 
 to \ref{A_eps_u_omega}, $V(u_i|\tilde{X}_i, z_i)=\sigma_u^2 I_i$, and
    we use a linear projection in the reduced form equation, i.e.
     $M_{\tau_i}Z_i a_1 + M_{\tau_i} \tilde{X}_i a_2$ after transformation, then  
    \begin{equation}
 \mathbb{E} \Big[(M_{\tau_i}Z_i a_1 - M_{\tau_i}\tilde{X}_i \tilde{\Gamma})'\Phi_{\tau \tilde{u},i}^{+}\
 (M_{\tau_i}Z_i a_1 - M_{\tau_i}\tilde{X}_i \tilde{\Gamma})] \leq \pi_o.  
\end{equation}
\noindent where $\tilde{\Gamma}=
\big(\mathbb{E}[\tilde{X}_i'M_{\tau_i}'\Phi_{\tau \tilde{u},i}^{+}M_{\tau_i}\tilde{X}_i]\big)^{-1}
\big(\mathbb{E}[\tilde{X}_i'M_{\tau_i}'\Phi_{\tau \tilde{u},i}^{+}M_{\tau_i}Z_i a_1]\big)$,
 $a_1$, $a_2$ represent the parameters of the linear projection, 
and $\Phi_{\tau \tilde{u},i}=\mathbb{E}[M_{\tau_i} \tilde{u}_i  \tilde{u}_i'
 M_{\tau_i}']$ is the variance-covariance matrix
  of the transformed misspecified reduced form error term.
\end{proposition}

Proposition \ref{Proposition_weak_IV} states that misspecification of the relationship
between the endogenous covariate and the instruments leads to a concentration 
parameter lower than or equal to the concentration parameter in the population $\pi_o$.  The equality 
holds if  $\tau g_o(\tilde{x}_{it}, z_{it})$ lies in the span of 
$[\tau \tilde{X}_{it}, \tau Z_{i}]$. Thus, opening the door to a situation where 
instruments could be weak even if the concentration parameter 
in the population is bounded away from 0.

\subsection{An orthogonal score  }\label{S_OrthogonalScore}

In this subsection, we derive an orthogonal score function.

\begin{proposition}
The score  associated to the moment condition \ref{non_orthogonal_score} is not Neyman
 orthogonal. 
\end{proposition}
\begin{proof}
 The Gateaux derivative of the score is equal to 
\begin{equation}
 \begin{split}
    D_h[\tau \tilde{g} ] = & \lim_{h \to 0} \mathbb{E}\Big[\frac{(M_{\tau_i}H_{h,i})'V_i^{-1}M_{\tau_i}(y_i - x_i \beta_o - (x_{1i}-(g_o+h\tilde{g})\rho_o))-(M_{\tau_i}H_i)'V_i^{-1}M_{\tau_i}(y_i - x_i \beta_o - (x_{1i}-g_o)\rho_o) }{h}\Big]\\    
 \end{split}
\end{equation}

\noindent with $H_{i,h}[x_{1i}, \tilde{X}_i, u_i - h\tilde{g}]$ since 
$x_{1i}-(g_o+h\tilde{g})=u_i - h\tilde{g}$, and we do not specify the arguments of $g$ for
simplicity knowing that they are equal to $\tilde{X}_i$ and $z_i$.  

The Gateaux derivative evaluated at $h=0$ is equal to $\mathbb{E}\Big[(M_{\tau_i}H_i)'V_i^{-1}M_{\tau_i}( \tilde{g}\rho_o)) \Big]$.
Note that it is not equal to 0 for all the $\beta_o$ rows except for $\rho_o$ row. 
More specifically, we can split  $D_0[\tau \tilde{g}] $ into $D_{\beta, 0}[\tau 
\tilde{g} ]$, and  $D_{\rho,0}[\tau 
\tilde{g} ]$. Such that $D_{\beta, 0}[\tau \tilde{g} ] =
\mathbb{E}\Big[(M_{\tau_i} [x_{1i}, \tilde{X}_i])'V_i^{-1}M_{\tau_i}( 
\tilde{g}\rho_o) \Big]  \neq 0$, and 
$D_{\rho, 0}[\tau \tilde{g} ] = \mathbb{E}\Big[(M_{\tau_i}
[u_i])'V_i^{-1}M_{\tau_i}( \tilde{g}\rho_o) - (M_{\tau_i}
[\tilde{g}])'V_i^{-1}M_{\tau_i}(\omega_i)  \Big] =0$. 
\end{proof}

By the Riesz representation theorem, the Riesz representer is
 $\rho_o\mathbb{E}[(M_{\tau_i}H_i)| \tilde{X}_i, z_i]$. 
By Proposition 4 of \cite{newey1994asymptotic}, the correction term is equal to
\begin{equation}
    \kappa = (M_{\tau_i}[g(\tilde{X}_{i},z_{i}), \tilde{X}_{i}, \textbf{0}])'V_i^{-1}M_{\tau_i} (x_{1i} - g(\tilde{X}_i, z_i)) \rho . 
\end{equation}
Then, adding $\kappa$ to $\phi(O_i;\theta, g)$  gives the orthogonal score
\begin{equation}\label{orthogonalScore}
    \psi(O_i;\theta, \tau g) =  (M_{\tau_i}H_{i,g})'V_i^{-1}M_{\tau_i}(y_i-x_i\beta
    -\rho(x_{1i}-g(\tilde{X}_i, z_i))) + (M_{\tau_i}[g(\tilde{X}_{i},z_{i}), 
    \tilde{X}_{i}, \textbf{0}])'V_i^{-1}M_{\tau_i} (x_{1i} - g(\tilde{X}_i, z_i)) \rho, 
\end{equation} 
\noindent with $H_{i,g} = [x_{1i}, \tilde{X}_i, x_{1i}-g(\tilde{X}_i,z_i)]$.
\begin{proposition}[Orthogonal score]
The score \ref{orthogonalScore} is Neyman Orthogonal.    
\end{proposition}

Then, the identifying moment condition is given by

\begin{equation}\label{Orthogonal_moment}
\mathbb{E}[\psi(O_i;\theta_0, \tau g_o) ] = 0.
\end{equation}  

\section{Orthogonal Super Learner Control Function estimation}\label{S_SLCF}

As described in the previous section,  we can write the transformed structural equation for observation $i$ at period $t$ as follows: 
\begin{equation} 
\tau y_{it}=\tau x_{1it}\beta_{1_o}+\tau \tilde{x}_{it}'\beta_{2_o}+\rho_o \tau u_{it} + \tau \omega_{it} , i\in \{1, ..., N\}, t \in \{t_a, ..., T_i\}.
\end{equation}
In this transformed equation, the parameters of interest are identified because $\tau x_{1it}$, $\tau \tilde{x}_{it}$ and $\tau u_{it}$ are uncorrelated with error term $\tau\omega_{it}$ under assumptions \ref{A_structural_eq} - \ref{A_eps_u_omega}.

In addition, it is clear that if $\tau u_{it}$ is observed,  we could estimate the
 three parameters $\beta_{1_o}$, $\beta_{2_o}$,  $\rho_o$  using the sample 
 counterparts of the population moment conditions \ref{Orthogonal_moment}.  But since $\tau u_{it}$ is not observed,  we propose the following estimation procedure:

    \noindent \textbf{First Step}: 
    
    \noindent We partition the set $\{1, 2,  ..., N\}$ in B subsets $S_1$, $S_2$, ..., $S_B$, and denote $n_{T,b} = \sum_{i \in S_b} T_i$. 
     This sample splitting is equal to the one proposed by \cite{EmmeneggerSJS2023} and it was developed without knowledge of it. Alternatively, if the cluster structure in the data is known (individuals exhibit correlation within clusters but not across them, and both cluster membership and the number of clusters are identified) the data should be partitioned accordingly.
\noindent For each $b \in\{1,\dots,B\}$ we estimate
  $\tau g_o(A_i) = \mathbb{E}[\tau x_{1it}|A_i]$, with $A_i$ the information set
  corresponding to the transformation $\tau$ (See Proposition \ref{Ident_nuisance}), using a super learner
   with partition $S_b^c= \{i \in S_g , g \neq b\}$, and we call
    the estimation $\widehat{\tau g_o}^{S_b^c}$. To simplify notation, 
    we write $\tau g_o(\tilde{x}_{it}, z_{it})$ and $\widehat{\tau g_o}^{S_b^c}(\tilde{x}_{it},
     z_{it})$ for $\tau g_o(A_i)$ and $\ \widehat{\tau g_o}^{S_b^c}(A_i)$, where $A_i$
      is the information set of Proposition \ref{Ident_nuisance}. Then,
       we obtain the residuals $\widehat{\tau u}_{it}^{S_b}=
        \tau x_{1it}-\widehat{\tau g_o}^{S_b^c}(\tilde{x}_{it},z_{it})$ for partition $S_b =\{i \in S_b\}$.

Finally,  the estimator of $\theta_o = [\beta_{1_o} \quad \beta_{2_o}' \quad  \rho_o]'$ for partition $S_b$ is the solution of
 the sample moment conditions such that: 
\begin{equation}\label{sample_moments}
\frac{1}{n_{T,b}}\sum_{i \in S_b}\psi(O_i;\hat{\theta}_b,\widehat{\tau g_o}^{S_b^c})=0,
\end{equation} 
\noindent where
\begin{multline*}
  \psi(O_i;\hat{\theta}_b,\widehat{\tau g_o}^{S_b^c}) =(M_{\tau_i}H^
{S_b}_{i,\widehat{\tau g_o}})'V_i^{-1}(M_{\tau_i}y_i-M_{\tau_i}X_i\hat{\beta}_b -\hat{\rho}_b(
M_{\tau_i}x_{1i}-\widehat{\tau g_o}^{S_b^c}(\tilde{X}_i, 
z_i))) \\
+ ([\widehat{\tau g_o}^{S_b^c}(\tilde{X}_{i},z_{i}),
M_{\tau_i}\tilde X_{i}, \textbf{0}])'V_i^{-1}(M_{\tau_i} x_{1i} - \widehat{\tau g_o}^{S_b^c}(\tilde{X}_i, 
z_i))
 \hat{\rho}_b,
\end{multline*} $\hat{\beta}_b = [\hat{\beta}_{1,b} \quad
 \hat{\beta}_{2,b}']'$, $M_{\tau_i}\tilde{ H}_{i,\widehat{\tau g_o}}^{S_b}=[ M_{\tau_i} x_{1i} 
\quad   M_{\tau_i} \tilde{X}_{i} \quad \widehat{\tau u}_{i}^{S_b}]$, 
$\widehat{\tau u}_{i}^{S_b}$ is the vector stacking the residuals obtained
 in the first stage for partition $S_b$,  $\tilde{X}_i$ 
 a $T_i\times (K-1)$ matrix collecting all time observations of $\tilde{x}_{it}$.

This is equivalent to estimating $\beta_o$ by performing instrumental
variables estimation with $[\widehat{\tau g_o}^{S_b^c}, 
M_{\tau_i} \tilde X_i]$ as instruments and weighting 
matrix 
$V_i^{-1}$ (See Proposition \ref{Proposition_Comparison}). 
Importantly, we can set $V_i=I_i$ which is the most efficient if 
  $T = 2$, and $\epsilon_{it}$ is homoskedastic and uncorrelated across both $i$ 
  and $t$.  

    \noindent \textbf{Second Step}: 

    We average $\hat{\theta}_b$ to obtain the cross-fitting estimator of $\theta_o$ as:

    $$\hat{\theta}_o=\frac{1}{B}\sum_{b=1}^B \hat{\theta}_b.$$

    Thus, the cross-fitting estimator of our parameter of interest $\beta_{1o}$ is: 

     $$\hat{\beta}_{1o}=\frac{1}{B}\sum_{b=1}^B \hat{\beta}_{1,b}.$$

     In order to deal with the dependency of the estimator $\hat{\beta}_{1o}$ on
      the particular sample split, we 
      follow \cite{Chernozhukov2018EJ} by proposing to repeat the procedure $SS$ 
      times.  The estimates are aggregated either by using the mean or the median. 
      In addition, a correction term is added to the variance estimator of 
      $\hat{\theta}_{o}$ (See Subsection \ref{Including_sample_splitting_uncertainty}).

\begin{algorithm}
    \caption{Super Learner Control Function estimation}\label{Algortithm_CF1}
    \begin{algorithmic}
        \STATE  \textbf{Input:} Panel data sets $\{y_i,X_{i},z_{i}\}_{i\in \{1, 2,  ..., N\}}$ from model satisfying assumptions \ref{A_data} to \ref{A_eps_u_omega}, a natural number B.
        \STATE \textbf{Output:} An estimator of the parameter of interest $\beta_{1o}$.
        \FOR{$ss \in \{ 1, \dots, SS\} $}
       \STATE Split the individual index set $\{1, 2,  ..., N\}$ into B sets $S_1, S_2,...,S_B$. 
        \FOR{$b \in \{1, 2,  ..., B\}$}           
                \STATE Estimate the conditional expectation $\mathbb{E}[\tau x_{1it}|A_i]$ using a super learner and data corresponding to $S_b^c$.         
                 \STATE Predict $\mathbb{E}[\tau x_{1it}|A_i]$ for        $i \in S_b$.
                \STATE Compute the residual $\widehat{ \tau u}_{it}$ for $i \in S_b$ using the prediction of $\mathbb{E}[\tau x_{1it}|A_i]$ of previous step. 
                \STATE Estimate $\theta_{o} = [\beta_{1o}, \beta_{2o}', \rho_o]'$ solving the sample moment  conditions \ref{sample_moments} for $i \in S_b$.                                        
        \ENDFOR
        \STATE Compute $\hat{\theta}_{ss}=\frac{\sum_1^B \widehat{\theta_{b,ss}}}{B}$.
        \ENDFOR
        \STATE  Compute $\hat{\theta}_{o}=\frac{\sum_1^{SS} \widehat{\theta_{ss}}}{SS},$ or the median. 
    \end{algorithmic}
    \end{algorithm}
\begin{example_un}{First-difference transformation (continued):}

As described in the previous section,  we can write the first-differenced structural equation for observation $i$ at period $t$ as follows: 
\begin{equation} 
\Delta y_{it}=\Delta x_{1it}\beta_{1_o}+\Delta \tilde{x}_{it}'\beta_{2_o}+\rho_o \Delta u_{it} + \Delta \omega_{it} , i\in \{1, ..., N\}, t \in \{2, ..., T_i\}.
\end{equation}
\sloppy
First, we partition the set $\{1, 2,  ..., N\}$ in B subsets $S_1$, $S_2$, ..., $S_B$, 
and for each $b \in \{1, \dots, B\}$ we estimate $\mathbb{E}[\Delta x_{1it}|\tilde{x}_{it},
\tilde{x}_{it-1},z_{it},z_{it-1}]$ using a super learner with partition $S_b^c$. 
 Then, we get the residuals $\widehat{\Delta u}_{it}^{S_b}=\Delta x_{1it}-
 \widehat{\Delta g_o}^{S_b^c}(\tilde{x}_{it},z_{it})$ for partition $S_b$.  
We obtain the estimator of $\theta_o$ for partition $S_b$ as the solution
 of the sample moment conditions such that: 
\begin{equation}\label{sample_moments_ex}
\frac{1}{n_{T,b}}\sum_{i \in S_b}\psi(O_i;\hat{\theta}_b,\widehat{\Delta g}_o^{S_b^c})=0,
\end{equation} 
\noindent where 
\begin{multline*}
\psi(O_i;\hat{\theta}_b,\widehat{\Delta g}_o^{S_b^c}) = (D_i H^
{S_b^c}_{i,\widehat{\tau g_o}})'V_i^{-1}(D_iy_i-D_i X_i\hat{\beta}_b -\hat{\rho}_b(D_ix_{1i}-
\widehat{\Delta g_o}^{S_b^c}(\tilde{X}_i, 
z_i))) \\
+ ([\widehat{\Delta g_o}^{S_b^c}(\tilde{X}_{i},z_{i}),
D_i\tilde X_{i}, \textbf{0}])'V_i^{-1} (D_i x_{1i} - 
\widehat{\Delta g_o}^{S_b^c}(\tilde{X}_i, 
z_i))
 \hat{\rho}_b,
\end{multline*}
    
    $\hat{\beta}_b = [\hat{\beta}_{1,b} \quad
 \hat{\beta}_{2,b}']'$, $D_i\tilde{ H}_{i,\widehat{\tau g_o}}^{S_b}=[  D_i x_{1i} 
\quad  D_i\tilde{X}_{i} \quad \widehat{\Delta u}_{i}^{S_b}]$,
 and $\widehat{\Delta u}_{i}^{S_b}$
  are the residuals obtained in the first stage for partition $S_b$. 

   Finally, since $\hat{\beta}_{1o,b} \in \hat{\beta}_{o,b}$ we estimate
    our parameter of interest $\beta_{1o}$ using: 

     $$\hat{\beta}_{1o}=\frac{1}{B}\sum_{b=1}^B \hat{\beta}_{1,b}.$$

\begin{algorithm}
    \caption{Super Learner Control Function Estimation: First-difference transformation }\label{Algortithm_CF1_ex}
    \begin{algorithmic}
        \STATE  \textbf{Input:} Panel data sets $\{y_i,X_{i},z_{i}\}_{i\in \{1, 2,  ..., N\}}$ from model satisfying assumptions \ref{A_data} to \ref{A_eps_u_omega}, a natural number B.
        \STATE \textbf{Output:} An estimator of the parameter of interest $\beta_{1o}$.
        \FOR{$ss \in \{1, 2,  ..., SS\}$}            
        \STATE Split the individual index set $\{1, 2,  \dots, N\}$ into B sets $S_1, S_2,...,S_B$. 
        \FOR{$b \in \{1, 2,  ..., B\}$}           
                \STATE Estimate the conditional expectation $\mathbb{E}[\Delta x_{1it}|\tilde{x}_{it},\tilde{x}_{it-1},z_{it},z_{it-1}]$ using a super learner and data corresponding to $S_b^c$.        
                 \STATE Predict $\mathbb{E}[\Delta x_{1it}|\tilde{x}_{it},\tilde{x}_{it-1},z_{it},z_{it-1}]$ for        $i \in S_b$.
                \STATE Compute the residual $\widehat{\Delta u}_{it}$ for $i \in S_b$ using the prediction of $\mathbb{E}[\Delta x_{1it}|\tilde{x}_{it},\tilde{x}_{it-1},z_{it},z_{it-1}]$ of previous step. 
                \STATE Estimate $\theta_{o} = [\beta_{1o}, 
                \beta_{2o}', \rho_o]'$ solving the sample moment  conditions 
                \ref{sample_moments_ex} for $i \in S_b$.                        
        \ENDFOR
        \STATE Compute $\hat{\theta}_{ss}=\frac{\sum_1^B \widehat{\theta_{b,ss}}}{B}$.        
        \ENDFOR
        \STATE Compute $\hat{\theta}_{o}=\frac{\sum_1^{SS} \widehat{\theta_{ss}}}{SS}$, or the median.
    \end{algorithmic}
\end{algorithm}
\end{example_un}

\section{Large sample properties of the Orthogonal SLCF estimator}\label{S_Properties}
\begin{theorem} {Consistency of the Orthogonal SLCF estimator}\label{T_Consistency} 

   \noindent If $\{O_i = (y_i, X_i, z_i) \}_{i\in [N]}$ satisfy assumptions
    \ref{A_data} to \ref{A_eps_u_omega}, assumptions of subsection 
    \ref{Annex_assumptions} in the Annex hold, i) $\Omega_i$ is a positive 
    semi-definite matrix and $\Omega_i\mathbb{E}[\psi(O_i; \theta, \tau g_o)]=0$ 
    only if $\theta = \theta_o$, ii) $\theta \in \Theta$ with $\Theta$ 
    compact, iii) $\psi(O_i; \theta, \tau g_o)$ is continuous at each
    $\theta \in \Theta$, iv) the function class $\mathcal{F}=\{\psi(O_i; \theta, 
    \tau  g_o)\}$ satisfies a uniform integrability condition
     $\mathbb{E}[\operatorname*{sup_{ \theta \in \Theta}} \left \| \psi(O_i; 
     \theta, \tau g_o)\right \|] < \infty$,  then $\hat{\theta}_o
      \overset{p}{\to} \theta_o$.
\end{theorem}
The proof is presented in the appendix (Subsection \ref{Consistency_proof}). 

\begin{theorem}{Asymptotic normality of the Orthogonal SLCF estimator 
     of $\beta_{o}$}\label{Theorem_AN}

    \noindent If $\{O_i = (y_i, X_i, z_i) \}_{i\in [N]}$ satisfy assumptions \ref{A_data} to \ref{A_eps_u_omega}, assumptions of subsection \ref{Annex_assumptions} in the Annex hold, 
    then, 

\begin{equation}
    \sqrt{N_T} (\hat{\beta}_{o} - \beta_{o}) \xrightarrow{d}  \mathcal{N}(0, \Sigma_o),
\end{equation}

with $\Sigma_o = \mathbb{E}_P[\tilde{Z}_i'V_i^{-1}M_{\tau_i}X_i]^{-1}
 \mathbb{E}_{P} [ \tilde{Z}_{i}' V_i^{-1} M_{\tau_i} \epsilon_i \epsilon_i' M_{\tau_i}' V_i^{-1} 
\tilde{Z}_i] (\mathbb{E}_P[\tilde{Z}_i'V_i^{-1}M_{\tau_i}X_i] ^{-1})'$, 
and $\tilde{Z}=[\tau g_o(\tilde{X}_i, z_i) \quad M_{\tau_i}\tilde{X}_i]$.

\end{theorem}
The proof is presented in the appendix (Subsection \ref{Proof_AsymptoticNormality}). 
\begin{remark}
    The asymptotic distribution of the orthogonal SLCF estimator of the parameter $\rho_o$ is
     left for later study. The reason is that the $\rho_o$ has no correction in the orthogonal score derivation. 
\end{remark}
\subsection{Estimator of the variance-covariance matrix}
The estimator of the variance-covariance matrix of the SLCF estimator of $\beta_o$ is given by:

\begin{equation}\label{Var_estimator}
    \hat{\Sigma}_o = \hat{J}_o^{-1} 
    \frac{1}{B} \sum_{b=1}^B  \mathbb{E}_{n,b} [ \widehat{\tilde{Z}'}_{i}^{S_b} 
    V^{-1}_i 
    ( M_{\tau_i} y_{i}^{S_b} 
    -M_{\tau_i}  X_{i}^{S_b} \hat{\beta}_o )(M_{\tau_i} 
    y_{i}^{S_b} -
    M_{\tau_i} X_{i}^{S_b} \hat{\beta}_o )'V^{-1}_i 
    \hat{\tilde{Z}}_{i}^{S_b}  ]  (\hat{J}_o^{-1})' ,  
\end{equation}

where:

\begin{equation}\label{J}
    \hat{J}_o = \frac{1}{B} \sum_{b = 1}^B \mathbb{E}_{n,b} 
    [\widehat{\tilde{Z}'}_{i}^{S_b}V^{-1}_i M_{\tau_i} X_{i}^{S_b}].
\end{equation}

\begin{remark}\label{Within_correction}
In the case of Within transformation and $\epsilon_{it}$ are identical and 
independently distributed, there is need of adjusting the variance-covariance matrix
by multiplying by the factor $\frac{NT-K}{NT-N-K}$. The reason is that this provides 
valid statistical inference after wiping out the individual fixed effects (\citealt{ding2021frisch}).
Otherwise, the cluster-robust variance-covariance matrix of the Within estimator 
does not need any correction (\citealt{ding2021frisch}, \citealt{colin2015practitioner}).
\end{remark}
\subsubsection{Including sample splitting uncertainty}\label{Including_sample_splitting_uncertainty}

In order to take into account the sample splitting uncertainty, we can repeat the 
cross-fitting process several times as explained in Section \ref{S_SLCF} 
and algorithm \ref{Algortithm_CF1}.  We also need to add correction term to 
the variance-covariance estimator $\hat{\Sigma}_o$ given by:

\begin{equation}\label{SS_correction}
    C = \frac{1}{SS}\sum_{ss = 1}^{SS}(\hat{\beta}_{ss} - \hat{\beta}_o)(\hat{\beta}_{ss} - \hat{\beta}_o)', 
\end{equation}

with $\hat{\beta}_o = \frac{1}{SS} \sum_{ss = 1}^{SS}\hat{\beta}_{ss}$.

It is possible to use the median instead of the mean (\citealt{Chernozhukov2018EJ}).
\begin{algorithm}\label{Algortithm_VarCov}
    \caption{Estimator of the variance-covariance of the Super Learner Control Function estimator}\label{your_label}
    \begin{algorithmic}
        \STATE  \textbf{Input:} Panel data sets $\{y_i,X_{i},z_{i}\}_{i\in \{1, 2,  ..., N\}}$ 
        from model satisfying assumptions \ref{A_data} to \ref{A_eps_u_omega}, a natural number B. 
        The selected index set $\{1, 2,  ..., N\}$ for each set $S_1, S_2,...,S_B$ across all samples splits (SS) obtained 
        in algorithm \ref{Algortithm_CF1}. The estimated $\hat{\theta}_{ss}$ obtained in algorithm \ref{Algortithm_CF1}.
        \STATE \textbf{Output:} An estimator of the variance-covariance matrix of $\beta_{1o}$.                
        \FOR{$ss \in \{1, 2, ..., SS \}$}
        \FOR{$b \in \{1, 2,  ..., B\}$}           
                \STATE Estimate $\hat{J}_{o,b} = \mathbb{E}_{n,b} 
                [\tilde{Z'}_{i}^{S_b}V^{-1}_i M_{\tau_i}X_{i}^{S_b}]$.
                \STATE Estimate $ \mathbb{E}_{n,b} [ \tilde{Z'}_{i}^{S_b} 
    V^{-1}_i 
    ( M_{\tau_i} y_{i}^{S_b} 
    - M_{\tau_i} X_{i}^{S_b} \hat{\beta}_o )(M_{\tau_i} 
    y_{i}^{S_b} -
    M_{\tau_i} X_{i}^{S_b} \hat{\beta}_o )'V^{-1}_i 
    \tilde Z_{i}^{S_b} ]$.                 
        \ENDFOR
        \STATE Compute $\hat{\Sigma}_s.$                
        \ENDFOR
        \STATE  Compute $C=\frac{1}{SS}\sum_{ss = 1}^{SS}(\hat{\beta}_{ss} - 
        \hat{\beta}_o)(\hat{\beta}_{ss} - \hat{\beta}_o)'.$         
        \STATE Compute $\hat{\Sigma}_o = \frac{1}{SS} \sum_{ss=1}^{SS} 
        (\hat{\Sigma}_s) + C.$

    \end{algorithmic}
\end{algorithm}

\newpage
\section{Comparison of the Orthogonal SLCFE with 
the plug-in IV estimator}\label{S_Comparison}

\begin{proposition}[The Orthogonal Control Function estimator is equivalent to the plug-in IV estimator]\label{Proposition_Comparison}
For each fold $b$, let $\widehat{\tilde{Z}}:=[\widehat{\tau g_o}^{S_b^c}(\tilde X_i,z_i),\ M_{\tau_i}\tilde X_i]$ and $\hat u_i:=\widehat{\tau u}_i^{S_b}$ for $i\in S_b$, and suppose $\sum_{i\in S_b}\widehat{\tilde{Z}}'V_i^{-1}M_{\tau_i}X_i$ is nonsingular and $\sum_{i\in S_b}\hat u_i'V_i^{-1}\hat u_i>0$. Then $\hat\theta_b=(\hat\beta_b',\hat\rho_b)'$ solves the sample moment conditions (\ref{sample_moments}) if and only if
\begin{align}
\hat\beta_b&=\Big(\sum_{i\in S_b}\widehat{\tilde{Z}}'V_i^{-1}M_{\tau_i}X_i\Big)^{-1}\sum_{i\in S_b}\widehat{\tilde{Z}}'V_i^{-1}M_{\tau_i}y_i, \label{P_IV_beta}\\
\hat\rho_b&=\frac{\sum_{i\in S_b}\hat u_i'V_i^{-1}\big(M_{\tau_i}y_i-M_{\tau_i}X_i\hat\beta_b\big)}{\sum_{i\in S_b}\hat u_i'V_i^{-1}\hat u_i}. \label{P_IV_rho}
\end{align}
That is, the orthogonalized control function estimator of $\beta_o$ is the instrumental variables estimator with the learned instrument $\widehat{\tau g_o}^{S_b^c}$ (and $M_{\tau_i}\tilde X_i$ as its own instrument) and weighting matrix $V_i^{-1}$.
\end{proposition}
 
\begin{proof}
Write $r_i(\beta):=M_{\tau_i}y_i-M_{\tau_i}X_i\beta$ and $\sum$ for $\sum_{i\in S_b}$. By definition of the residual, $M_{\tau_i}x_{1i}-\widehat{\tau g_o}^{S_b^c}(\tilde X_i,z_i)=\hat u_i$, hence
\begin{equation}\label{P_id}
M_{\tau_i}X_i-\widehat{\tilde{Z}}=[\hat u_i,\ 0]=\hat u_ie_1',\qquad e_1=(1,0,\dots,0)'.
\end{equation}
Summing the score over $i\in S_b$, the moment conditions (\ref{sample_moments}) are
\begin{align}
&\textstyle\sum (M_{\tau_i}X_i)'V_i^{-1}\big(r_i(\beta)-\rho\hat u_i\big)+\rho\sum\widehat{\tilde{Z}}'V_i^{-1}\hat u_i=0 &&(\beta\text{ rows}),\label{P_b}\\
&\textstyle\sum\hat u_i'V_i^{-1}\big(r_i(\beta)-\rho\hat u_i\big)=0 &&(\rho\text{ row}).\label{P_r}
\end{align}
Collecting the terms in $\rho$ in (\ref{P_b}) and using (\ref{P_id}),
\[
\textstyle-\rho\sum(M_{\tau_i}X_i-\widehat{\tilde{Z}})'V_i^{-1}\hat u_i=-\rho\, e_1\sum\hat u_i'V_i^{-1}\hat u_i ,
\]
so (\ref{P_b}) is equivalent to
\begin{equation}\label{P_b2}
\textstyle\sum(M_{\tau_i}X_i)'V_i^{-1}r_i(\beta)=e_1\,\rho\sum\hat u_i'V_i^{-1}\hat u_i .
\end{equation}
(For the rows of $\beta_2$ the two terms in $\rho$ cancel exactly; for the row of $\beta_1$ they reduce to $\rho\sum\hat u_i'V_i^{-1}\hat u_i$.) By (\ref{P_r}), $\rho\sum\hat u_i'V_i^{-1}\hat u_i=\sum\hat u_i'V_i^{-1}r_i(\beta)$. Substituting in (\ref{P_b2}) and using (\ref{P_id}) again,
\[
\textstyle\sum\big(M_{\tau_i}X_i-\hat u_ie_1'\big)'V_i^{-1}r_i(\beta)=\sum\widehat{\tilde{Z}}'V_i^{-1}\big(M_{\tau_i}y_i-M_{\tau_i}X_i\beta\big)=0,
\]
which, under the nonsingularity condition, is equivalent to (\ref{P_IV_beta}); and (\ref{P_r}) is (\ref{P_IV_rho}). Each step is an equivalence (adding zero, regrouping, and replacing one side of (\ref{P_r}) by the other), so the converse also holds. \end{proof}
 
\begin{remark}
The proposition is exact algebra in every sample and for any
 first-stage estimator; it uses neither conditional expectations nor 
 asymptotics. Cross-fitting is not needed for the equivalence; it is 
 needed for the asymptotic properties of the estimator (Theorems~\ref{T_Consistency}, \ref{Theorem_AN}).
  The learned-instrument interpretation connects the estimator with the 
  literature on estimated optimal instruments 
  \citep{newey1990efficient,belloni2012sparse} and 
  with machine-learning first stages in instrumental variables \citep{chen2020mostly}.
\end{remark}

\subsection{The Naive Plug-in 2SLS
 estimator and Plug-in IV estimation are not equivalent to
Naive SLCFE}
The naive Super Learner Control function estimator based 
on the moment condition \ref{non_orthogonal_score} is not
 numerically equivalent to a 2SLS estimator that
 plugs in the estimated transformed nuisance parameter
  $\tau g_o(\tilde{x}_{it}, z_{it})$.  
 Without loss of generality, we assume that exogenous covariates are not present (model 
 \ref{structural_equation_toy} to \ref{relationship_errors_toy}), and $V_{i} = I_{i}$.  
 As described in Section \ref{S_Identification}, one could be tempted to estimate the parameter of
interest using the naive 2SLS estimator or the naive Control Function estimator

\begin{equation}
    \hat{\beta}_{2SLS, 1_o} = \sum_b^B \frac{1}{B}(\widehat{\tau g_o}^{S_b^c}( z)'
         \widehat{\tau g_o}^{S_b^c}( z))^{-1} 
    (\sum_{i \in S_b}\widehat{\tau g_o}^{S_b^c}( z)' M_{\tau} y^{S_b}).
\end{equation}

\begin{equation}
    \hat{\beta}_{Naive\_SLCFE, 1_o} = \sum_b^B \Big( (M_{\tau}x_{1}^b)' 
    \hat{A}_{\tau u} M_{\tau} x_{1}^{S_b} \Big)^{-1}
     \Big( (M_{\tau} x_{1}^{S_b}) ' \hat{A}_{\tau u}  M_{\tau} y^{S_b}\Big),
\end{equation}
\noindent where $\hat{A}_{\tau u} = I_{N_{b,T_{a}}} - \widehat{\tau u} 
(  \widehat{\tau u}'\widehat{\tau u})^{-1} \widehat{\tau u}'$ with 
$N_{T_{a}} = \sum_{i\in {b}} T_{ai}$, and $T_{ai}$ the number of observations left
after the transformation, the exponent $S_b$ denotes the observations
 belonging to fold ${b}$ but for ease of notation we strip them.

\begin{proposition}[The Naive Control Function estimator and the plug-in 2SLS estimators 
    are not equivalent]\label{plugin2SLSvsNaiveCF}
    The naive plug-in
     2SLS estimation is 
    numerically equivalent to the naive SLCFE 
    if and only if $\widehat{\tau u}'\widehat{\tau g_o}( z)=0$.
     A sufficient condition for the later is $\widehat{\tau g_o} \in span(\tau z)$, thus 
     obtained by least-squares estimation on the same observations.

\end{proposition}

\begin{proof}
    Assuming that $B = 1$, we can write the estimators as:

\begin{equation}
    \hat{\beta}_{2SLS, 1_o} = \Big(\widehat{\tau g_o}( z)'\widehat{\tau g_o}( z)\Big)^{-1}
     \Big(\widehat{\tau g_o}( z)' M_{\tau} y\Big),
\end{equation}

\begin{equation}
    \hat{\beta}_{Naive\_SLCFE, 1_o} = \Big( (M_{\tau}x_{1})'
     \hat{A}_{\tau u} M_{\tau} x_{1} \Big)^{-1}
     \Big((M_{\tau} x_{1}) ' \hat{A}_{\tau u}  M_{\tau} y\Big),
\end{equation}
 For simplicity, let \(q=\widehat{\tau g_o}(z)\) and \(u=\widehat{\tau u}\). Assuming the first-stage decomposition \(M_\tau x_1=q+u\), the two estimators can be expressed as

\[
\widehat{\beta}_{2SLS,1_o}
=(q'q)^{-1}q'M_\tau y,
\]

and

\[
\widehat{\beta}_{Naive\_SLCFE,1_o}
=\big((M_\tau x_1)'A_u M_\tau x_1\big)^{-1}
(M_\tau x_1)'A_u M_\tau y,
\]

where

\[
A_u=I_{N_{T_a}}-u(u'u)^{-1}u'.
\]

Since \(A_u u=0\) and \(M_\tau x_1=q+u\), we obtain

\[
A_u M_\tau x_1=A_u q.
\]

Thus, the two estimators are generally different because \(A_u q\neq q\) unless \(q'u=0\).

In particular, they are numerically equivalent if

\[
\widehat{\tau g_o}(z)'\widehat{\tau u}=0.
\]

A sufficient condition for this orthogonality is that

\[
\widehat{\tau g_o}(z)\in\operatorname{col}(\tau z),
\]

provided that \(\widehat{\tau u}\) is
 obtained as the residual from projecting 
 \(M_\tau x_1\) onto \(\operatorname{col}(\tau z)\).
  This condition is satisfied, for example, when 
  \(\tau g_o(.)\) is estimated using least-squares on the same observations.

\end{proof}

\begin{remark}[The plug-in 2SLS estimator is not $\sqrt{N_T}$ consistent]
       If $\{O_i = (y_i, X_i, z_i) \}_{i\in [N]}$ satisfy assumptions
    \ref{A_data} to \ref{A_eps_u_omega}, assumptions \ref{A_ThetaCompact} - \ref{A_several} of subsection 
    \ref{Annex_assumptions} in the Annex hold, the Naive plug-in 2SLS estimator is not $\sqrt{N_T}$-consistent.

    The naive plug-in 2SLS estimator with cross-fitting is given by:
\begin{equation}
    \hat{\beta}_{2SLS,1_o} = \frac{1}{B}\sum_b^B\Big( \widehat{\tau g_o}^{S_b^c}( z)
    '\widehat{\tau g_o}^{S_b^c}( z)\Big)^{-1}
     \Big(\widehat{\tau g_o}^{S_b^c}( z)' M_{\tau} y\Big).
\end{equation}
Replacing $M_{\tau} y$ by $\tau g_o(z) \beta_{1_o} + M_{\tau}u \beta_{1_o} +
 M_{\tau } 
\varepsilon$, adding and subtracting $\widehat {\tau g_o}^{S_b^c}(z)\beta_{1_o}$,
 it can be written as: 

\begin{equation}
    \hat{\beta}_{2SLS,1_o} = \beta_{1_o} + \frac{1}{B}\sum_b^B\Big( \widehat{\tau 
    g_o}^{S_b^c}( z)'\widehat{\tau g_o}^{S_b^c}( z)\Big)^{-1}
     \Big(\widehat{\tau g_o}^{S_b^c}( z)'( -\xi^{S_b^c} \beta_{1,o} + M_{\tau }u
     \beta_{1_o} +  M_{\tau} \varepsilon)\Big),
\end{equation}

with $-\xi^{S_b^c} = \tau g_o(z) - \widehat{ \tau g_o} (z)^{S_b^c}$. 

Heuristically, the second term does not converge to zero at rate faster than $\sqrt{N_T}$ 
because  $\tau g_o(z_i)$, unlike $M_{\tau_i}u_i$
 and $M_{\tau_i}\varepsilon_i$, does not have 
 conditional mean zero given $z_i$, so the
  term $\sum_i\tau g_o(z_i)'\xi_i$ has
   non-zero mean and, after scaling by $\sqrt{N_T}$, 
   diverges \citep[Section~1.1]{Chernozhukov2018EJ}.
Thus, the naive plug-in estimator presents a first order bias term that does 
not vanish using sample-splitting after scaling by $\sqrt{N_T}$ when the first-stage 
estimator converges more slowly than the parametric rate.  Thus, the naive plug-in 2SLS estimator
 is not $\sqrt{N_T}$-consistent.
\end{remark}

\begin{proposition}[The plug-in IV estimator is not equivalent to the Naive Control Function estimator]\label{pluginIVvsNaiveCF}
    The plug-in IV estimator is given by: 
\begin{equation}
    \hat{\beta}_{IV, 1_o} = \sum_b^B \frac{1}{B}( \widehat{\tau g_o}^{S_b^c}( z)'
     M_{\tau_i} x_{1})^{-1} (\widehat{\tau g_o}^{S_b^c}( z)' M_{\tau_i} y).
\end{equation}

Assuming that $B = 1$, we can write the estimators as:

\begin{equation}
    \hat{\beta}_{IV, 1_o} = \Big( \widehat{\tau g_o}( z)' M_{\tau_i} x_{1}\Big)^{-1}
     \Big(\widehat{\tau g_o}( z)' M_{\tau} y\Big),
\end{equation}

\begin{equation}
    \hat{\beta}_{Naive\_SLCFE, 1_o} = \Big( (M_{\tau}x_{1})' \hat{A}_{\tau u}
     M_{\tau} x_{1} \Big)^{-1}
     \Big((M_{\tau} x_{1}) ' \hat{A}_{\tau u}  M_{\tau} y\Big).
\end{equation}

\noindent 
If  $\widehat{\tau g_o( z)}'\widehat{\tau u}\neq 0$, then $ \hat{A}_{\tau u} M_{\tau} 
x_{1} \neq \widehat{\tau g_o}( z)$.  Thus, the two estimators are not equivalent unless 
 $\widehat{\tau g_o( z)}$ is learned using least-squares on the same observations without
 cross-fitting.
\end{proposition}
The proof is omitted as it is similar to the proof of Proposition \ref{plugin2SLSvsNaiveCF}.

\providecommand{\simfigdir}{figures_sim/}
\providecommand{\E}{\mathbb{E}}
\section{Monte Carlo simulations}\label{S_Simulation}
 
\subsection{Design}
We generate a balanced panel with $N=2500$ individuals and $T=2$ periods. For $i=1,\dots,N$ and $t=1,2$,
\begin{align*}
x_{1it}&=0.5\,\alpha_i+g_o(\tilde x_{it},z_{it})+u_{it},\\
y_{it}&=0.5\,\alpha_i+\tilde x_{it}\beta_{2_o}+x_{1it}\beta_{1_o}+\varepsilon_{it},\qquad \varepsilon_{it}=\rho_ou_{it}+\omega_{it},\\
g_o(\tilde x_{it},z_{it})&=a\,|z_{it}|-2\tanh(\tilde x_{it})+z_{it}/a,
\end{align*}
with $\alpha_i\sim U(-1,1)$, $\tilde x_{it}=\alpha_i+U(-2,2)$, $z_{it}\sim U(-2,2)$, $u_{it}\sim U(-2,2)$ and $\omega_{it}\sim U(-1,1)$, all independent, $\beta_{1_o}=\beta_{2_o}=1$ and $\rho_o=0.9$. The parameter $a\in\{1,\dots,10\}$ moves the first stage from a design with a strong linear component ($a=1$) to one in which the dependence of $x_1$ on $z$ is almost entirely nonlinear ($a=10$): the linear part $z/a$ shrinks and the nonlinear part $a|z|$ grows. Linear instruments therefore become weak as $a$ increases, while $g_o$ remains a strong instrument. We use the first-difference transformation; with $T=2$ it gives the same estimator as the within transformation. Each design is replicated $R=500$ times. Appendix~\ref{A_MC} reports a design with heterogeneous $\rho_i$.
 
\subsection{Estimators}
All estimators are computed on first-differenced data with regressors $(\Delta x_{1it},\Delta\tilde x_{it})$. We report results for $\beta_{1_o}$.
\begin{itemize}
\item \textbf{FD-OLS}: OLS of $\Delta y$ on $(\Delta x_1,\Delta\tilde x)$.
\item \textbf{FD-2SLS}: 2SLS with instruments $(\Delta z,\Delta\tilde x)$.
\item \textbf{FD-2SLS poly}: 2SLS with instruments (orthogonal polynomials of degree 5 in $\Delta z$, $\Delta\tilde x$).
\item \textbf{Oracle plug-in IV}: IV with instruments $(\Delta g_o,\Delta\tilde x)$, using the true $g_o$. This estimator is infeasible and serves as a benchmark.
\item \textbf{Orthogonal SLCF}: the estimator of Algorithm~\ref{Algortithm_CF1_ex}. By Proposition~\ref{Proposition_Comparison}, it equals the cross-fitted IV estimator with instruments $(\widehat{\Delta g},\Delta\tilde x)$.
\item \textbf{Naive SLCF}: the control function estimator without the correction term, that is, OLS of $\Delta y$ on $(\Delta x_1,\Delta\tilde x,\widehat{\Delta u})$ with the same cross-fitted residuals and the usual OLS sandwich variance.
\item \textbf{Plug-in 2SLS}: OLS of $\Delta y$ on $(\widehat{\Delta g},\Delta\tilde x)$, cross-fitted.
\item \textbf{Plug-in IV no cross-fit.}: IV with instruments $(\widehat{\Delta g},\Delta\tilde x)$, where $\widehat{\Delta g}$ is fitted on the full sample, without cross-fitting.
\end{itemize}
The last three estimators isolate the ingredients of the Orthogonal SLCF estimator. Naive SLCF is cross-fitted but not orthogonal. Plug-in IV no cross-fit.\ is orthogonal but not cross-fitted. Plug-in 2SLS is cross-fitted but uses $\widehat{\Delta g}$ as a regressor instead of as an instrument.
 
\subsection{Implementation}
The super learner combines the sample mean, linear regression, 
a neural network with one hidden layer of 10 units and linear output,
and gradient boosting (\texttt{xgboost}). Its weights are chosen by
5-fold cross-validation with folds formed by individual. Its
features are $(z_{it},z_{it-1},\tilde x_{it},\tilde x_{it-1})$
for all $a$, the variables of the conditioning set in 
Algorithm~\ref{Algortithm_CF1_ex}. Cross-fitting uses $B=5$ 
folds formed by individual. Within each sample split, the
estimator is computed fold by fold and averaged over folds.
We repeat this for $SS=3$ sample splits and take the median
$\hat\beta_1=\operatorname{median}\{\hat\beta_{1,ss}:ss=1,\dots,SS\}$,
 with variance 
$\hat V_1=\operatorname{median}\big\{\hat V_{ss}+(\hat\beta_{1,ss}-
\hat\beta_1)^2:ss=1,\dots,SS\big\}$ 
(Subsection~\ref{Including_sample_splitting_uncertainty}), where 
$\hat V_{ss}$ is the variance estimator (\ref{Var_estimator})--(\ref{J}) of split $ss$ divided by $N_T$. All standard errors are cluster-robust at the individual level. With $R=500$, the Monte Carlo standard error of a coverage rate is at most $0.022$, and $0.010$ at $0.95$.
 
\subsection{Results}
Table~\ref{tab:mc_main} reports the median bias, the mean squared error (MSE) and the coverage of nominal 95\% confidence intervals. Figure~\ref{fig:mc} plots the bias and the coverage against $a$. First-stage diagnostics, interquartile ranges, confidence-interval lengths and further robustness checks are in Appendix~\ref{A_MC}.
 
The first stage behaves as intended. The cluster-robust $F$ statistic of $\Delta z$ falls from $1840.5$ at $a=1$ to $0.9$ at $a=10$ and is below 10 for $a\geq6$, while the out-of-fold $R^2$ of the super learner increases from $0.70$ to $0.95$ (Table~\ref{tab:mc_diag}).
 
FD-OLS has median bias between $0.383$ and $0.035$ and zero coverage; its MSE is nevertheless smaller than that of FD-2SLS for $a\geq5$. FD-2SLS has median bias of at most $0.023$ in absolute value, but its dispersion grows with $a$: its interquartile range increases from $0.026$ at $a=1$ to $0.256$ at $a=10$, and its coverage is at least $0.99$ for $a\geq7$. For $a\geq7$, its MSE is dominated by a few replications with very large errors, as expected for an exactly identified IV estimator with a weak instrument. FD-2SLS poly has smaller dispersion than FD-2SLS for large $a$, but median bias between $0.034$ and $0.059$ for $a\geq5$, the bias towards OLS expected with several weak instruments. Polynomials of $\Delta z$ do not help because, in first differences, they are correlated with $\Delta g_o$ only through the linear term (Appendix~\ref{A_MC}).
 
The Orthogonal SLCF estimator has median bias of at most $0.002$ in absolute value for all $a$ and coverage between $0.93$ and $0.97$, within two Monte Carlo standard errors of $0.95$. Its interquartile range and confidence-interval length are within $0.003$ of those of the Oracle plug-in IV estimator, and within $0.001$ for $a\geq4$. Relative to FD-2SLS, its interquartile range is about the same at $a=1$ ($0.024$ against $0.026$) and about 50 times smaller at $a=10$ ($0.005$ against $0.256$). Its MSE exceeds that of the Oracle plug-in IV estimator by at most 8\% (at $a=1$) and by at most 6\% for $a\geq2$. It is smaller than the MSE of FD-2SLS for every $a$, by a factor of 1.2 at $a=1$ and of more than 200 at $a=5$, and smaller than that of FD-2SLS poly by a factor of more than 100 for $a\geq5$. In this design, the Orthogonal SLCF estimator thus attains the precision of the infeasible estimator that uses the true $g_o$.
 
The three alternative estimators fail in different ways.
\begin{itemize}
\item \textbf{Naive SLCF} has median bias between $-0.014$ and $-0.039$, increasing in absolute value with $a$, and coverage that falls from $0.73$ at $a=1$ to $0.16$ at $a=10$. Its confidence intervals are too short: the OLS sandwich treats $\widehat{\Delta u}$ as data and ignores the effect of the first-stage estimation on the limit distribution, which is not zero because this score is not orthogonal.
\item \textbf{Plug-in 2SLS} has median bias between $0.018$ and $0.051$, increasing with $a$, and coverage that falls from $0.96$ at $a=1$ to $0.18$ at $a=10$. The estimation error $\widehat{\Delta g}-\Delta g_o$ enters the regressor, so the estimator is consistent only if $\widehat{\Delta g}$ is consistent for $\Delta g_o$. In the IV estimator, $\widehat{\Delta g}$ enters only as an instrument, and cross-fitting makes it uncorrelated with the structural error of the same individual, whether or not it is close to $\Delta g_o$.
\item \textbf{Plug-in IV no cross-fit.}\ has median bias of $0.069$ at $a=1$, decreasing to $0.008$ at $a=10$, and coverage between $0.07$ and $0.48$. When $\widehat{\Delta g}$ is fitted on the full sample, the prediction for individual $i$ depends on $\Delta x_{1i}$, hence on $\Delta u_i$, and is therefore correlated with $\Delta\varepsilon_i=\rho_o\Delta u_i+\Delta\omega_i$.
\end{itemize}
Their MSEs are between 5.6 and 400 times that of Orthogonal SLCF: Naive SLCF between $2.08\times10^{-3}$ and $4.49\times10^{-3}$, Plug-in 2SLS between $2.98\times10^{-3}$ and $7.61\times10^{-3}$, and Plug-in IV no cross-fit.\ between $0.094\times10^{-3}$ and $6.44\times10^{-3}$. Cross-fitting and the orthogonal (IV) form are therefore both needed. Dropping either gives a median bias between $0.35$ and $2.5$ Monte Carlo standard deviations of the estimator and, except for Plug-in 2SLS with $a\leq3$, coverage well below the nominal level.
 
\paragraph{Identification through functional form.} In this design, $g_o$ is nonlinear in the exogenous regressor $\tilde x$ as well as in $z$. The learned instrument can therefore also use the nonlinearity in $\tilde x$, which is excluded from the outcome equation only through the assumption that $\tilde x$ enters it linearly 
(\citealp{angrist2022machine}; \citealp{chen2020mostly}).
 
\begin{table}[htbp]\centering\small
\caption{Median bias, mean squared error and coverage of $\hat\beta_1$}\label{tab:mc_main}
\setlength{\tabcolsep}{4pt}
\begin{tabular}{rcccccccc}
\toprule
$a$ & FD-OLS & FD-2SLS & \begin{tabular}[b]{@{}c@{}}FD-2SLS\\poly\end{tabular} & \begin{tabular}[b]{@{}c@{}}Oracle\\plug-in IV\end{tabular} & \begin{tabular}[b]{@{}c@{}}Orthogonal\\SLCF\end{tabular} & \begin{tabular}[b]{@{}c@{}}Naive\\SLCF\end{tabular} & \begin{tabular}[b]{@{}c@{}}Plug-in\\2SLS\end{tabular} & \begin{tabular}[b]{@{}c@{}}Plug-in IV\\no cross-fit.\end{tabular} \\
\midrule
\multicolumn{9}{l}{\textit{A. Median bias}}\\
1 & 0.383 & 0.001 & 0.003 & 0.001 & $-0.000$ & $-0.014$ & 0.018 & 0.069 \\
2 & 0.383 & 0.002 & 0.007 & 0.000 & $-0.001$ & $-0.019$ & 0.022 & 0.065 \\
3 & 0.261 & 0.003 & 0.014 & $-0.000$ & $-0.001$ & $-0.020$ & 0.024 & 0.040 \\
4 & 0.174 & 0.003 & 0.022 & $-0.001$ & $-0.002$ & $-0.024$ & 0.028 & 0.028 \\
5 & 0.122 & 0.004 & 0.034 & 0.000 & 0.000 & $-0.027$ & 0.032 & 0.022 \\
6 & 0.090 & 0.001 & 0.039 & 0.001 & 0.001 & $-0.031$ & 0.037 & 0.018 \\
7 & 0.067 & 0.022 & 0.059 & 0.000 & $-0.000$ & $-0.033$ & 0.045 & 0.014 \\
8 & 0.053 & 0.009 & 0.046 & $-0.000$ & $-0.000$ & $-0.037$ & 0.048 & 0.011 \\
9 & 0.043 & 0.007 & 0.047 & 0.000 & 0.000 & $-0.038$ & 0.050 & 0.010 \\
10 & 0.035 & 0.023 & 0.036 & 0.000 & $-0.000$ & $-0.039$ & 0.051 & 0.008 \\
\midrule
\multicolumn{9}{l}{\textit{B. Mean squared error $\times 1000$}}\\
1 & 146 & 0.438 & 0.443 & 0.328 & 0.353 & 2.08 & 2.98 & 6.44 \\
2 & 146 & 1.68 & 1.68 & 0.317 & 0.337 & 2.61 & 3.39 & 4.91 \\
3 & 67.8 & 4.30 & 4.16 & 0.158 & 0.166 & 2.87 & 3.91 & 2.23 \\
4 & 30.2 & 7.95 & 6.64 & 0.106 & 0.111 & 3.58 & 4.50 & 1.01 \\
5 & 14.9 & 18.5 & 10.7 & 0.0759 & 0.0771 & 3.73 & 5.54 & 0.602 \\
6 & 8.04 & 51.2 & 16.7 & 0.0494 & 0.0503 & 4.11 & 6.72 & 0.429 \\
7 & 4.57 & 2,963 & 16.6 & 0.0349 & 0.0358 & 3.45 & 5.78 & 0.264 \\
8 & 2.78 & 6,134 & 20.4 & 0.0250 & 0.0254 & 4.49 & 7.61 & 0.170 \\
9 & 1.81 & 207,650 & 17.2 & 0.0222 & 0.0224 & 4.35 & 7.44 & 0.134 \\
10 & 1.22 & 2,296 & 18.4 & 0.0164 & 0.0166 & 3.70 & 6.62 & 0.0936 \\
\midrule
\multicolumn{9}{l}{\textit{C. Coverage of nominal 95\% confidence intervals}}\\
1 & 0.00 & 0.95 & 0.95 & 0.94 & 0.94 & 0.73 & 0.96 & 0.07 \\
2 & 0.00 & 0.96 & 0.95 & 0.95 & 0.96 & 0.67 & 0.95 & 0.10 \\
3 & 0.00 & 0.95 & 0.93 & 0.96 & 0.95 & 0.69 & 0.93 & 0.16 \\
4 & 0.00 & 0.96 & 0.95 & 0.95 & 0.95 & 0.57 & 0.83 & 0.22 \\
5 & 0.00 & 0.96 & 0.94 & 0.93 & 0.93 & 0.52 & 0.76 & 0.28 \\
6 & 0.00 & 0.98 & 0.94 & 0.93 & 0.93 & 0.42 & 0.62 & 0.27 \\
7 & 0.00 & 0.99 & 0.95 & 0.95 & 0.96 & 0.36 & 0.55 & 0.32 \\
8 & 0.00 & 0.99 & 0.96 & 0.96 & 0.97 & 0.28 & 0.40 & 0.40 \\
9 & 0.00 & 1.00 & 0.98 & 0.93 & 0.94 & 0.21 & 0.26 & 0.39 \\
10 & 0.00 & 1.00 & 0.99 & 0.94 & 0.94 & 0.16 & 0.18 & 0.48 \\
\bottomrule
\end{tabular}
\par\medskip\begin{minipage}{0.95\linewidth}\footnotesize 500 replications, $N=2500$, $T=2$, first differences, $\beta_{1_o}=1$. Panel A: median of $\hat\beta_1-\beta_{1_o}$. Panel B: mean of $(\hat\beta_1-\beta_{1_o})^2$, multiplied by 1000. FD-2SLS is exactly identified; for $a\geq7$ its Monte Carlo mean squared error is driven by a few replications with very large $|\hat\beta_1-\beta_{1_o}|$ (up to 256 at $a=9$) and should be read together with the median bias and the interquartile range (Table~\ref{tab:mc_iqr}). Panel C: Monte Carlo standard error at most 0.022, and 0.010 at 0.95.\end{minipage}
\end{table}
 
\begin{figure}[htbp]\centering
\includegraphics[width=0.85\linewidth]{ 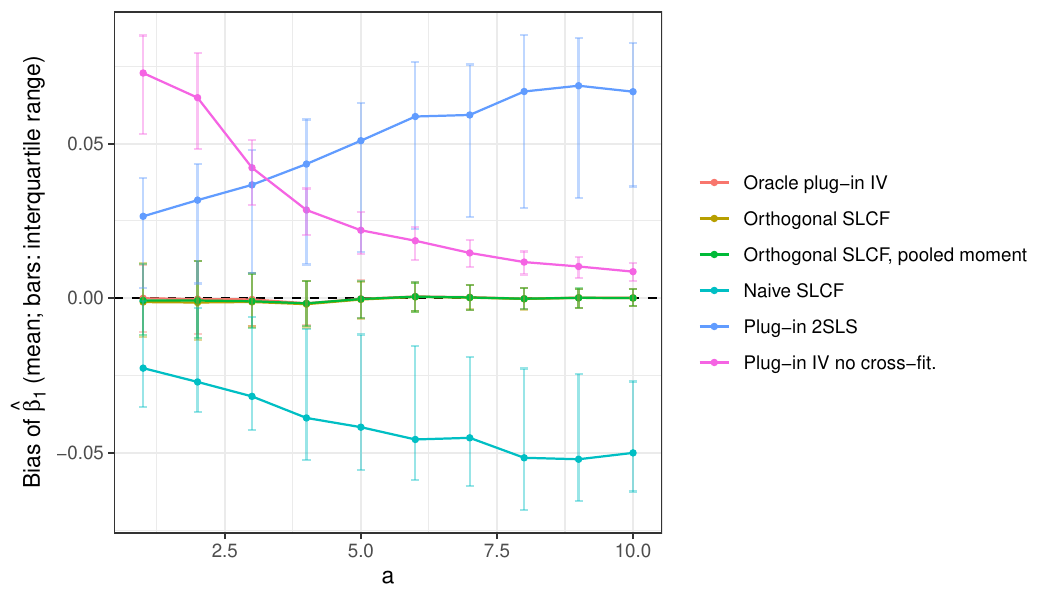}\\[2pt]
{\footnotesize (a) Mean bias of $\hat\beta_1$; bars are interquartile ranges of $\hat\beta_1-\beta_{1_o}$. FD-OLS, FD-2SLS and FD-2SLS poly are omitted for scale.}\\[8pt]
\includegraphics[width=0.85\linewidth]{ 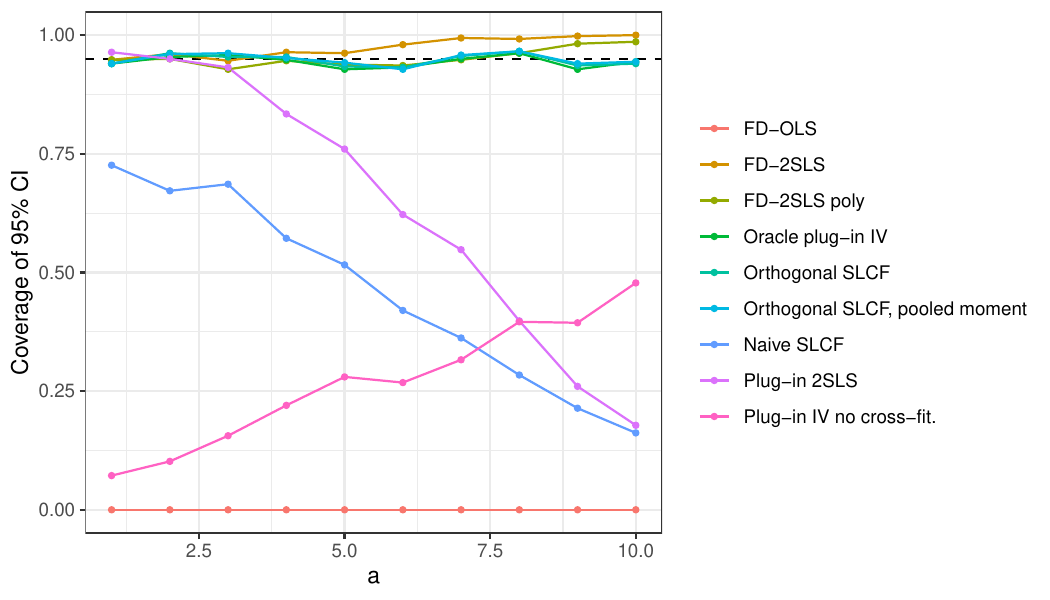}\\[2pt]
{\footnotesize (b) Coverage of nominal 95\% confidence intervals. The dashed line is $0.95$.}
\caption{Bias and coverage of $\hat\beta_1$ by $a$. 500 replications, $N=2500$, $T=2$, first differences.}\label{fig:mc}
\end{figure}

\section{Empirical application: the price elasticity of 
electricity demand}\label{S_EA}

\subsection{Design}
\citet{BurkeAbayasekara2018} estimate the short-run price
 elasticity of electricity demand with a three-dimensional 
 panel of US states $s$, sectors $j$ (residential, commercial,
  industrial) and years $t$. We write their model in levels as

\begin{align}
\ln E_{sjt}&=\beta_1\ln P_{sjt}+\alpha_{sj}+\lambda_{sj}t+
\mu_{jt}+\varepsilon_{sjt},\label{EA_y}\\
\ln P_{sjt}&=g\big(\ln P^{\mathrm{other}}_{sjt},\gamma_j,
\eta_t\big)+\alpha^x_{sj}+u_{sjt},\qquad \varepsilon_{sjt}=
\rho\,u_{sjt}+\omega_{sjt},\label{EA_x}
\end{align}
where $E_{sjt}$ is electricity sales, $P_{sjt}$ the average 
price (revenue divided by sales), $\alpha_{sj}$ and 
$\alpha^x_{sj}$ state-by-sector effects, $\lambda_{sj}t$ a
state-by-sector trend, and $\mu_{jt}$ sector-by-year effects; 
$\gamma_j$ and $\eta_t$ denote sector and year indicators. 
The price is endogenous: it is set jointly with demand, and an
average price falls when sales rise if part of the bill is 
fixed. The instrument is $\ln P^{\mathrm{other}}_{sjt}$, the
log average price paid by the other sectors of the same
state in the same year (for the residential sector: the
commercial, industrial and transportation sectors). It 
moves with the costs of supplying electricity in the state.
Both equations contain state-by-sector effects 
($\alpha_{sj}$, $\alpha^x_{sj}$, removed by first
differences) and sector-by-year effects ($\mu_{jt}$ 
in the outcome equation; $\gamma_j$, $\eta_t$ and their
interactions in the first stage). These absorb shocks
that are constant within a state and sector, and
shocks common to all states in a sector and year.
Identification requires that the remaining variation
in the instrument, which is state-specific and time-varying, is uncorrelated with the demand shock of sector $j$. 
The threat is that the included fixed effects (state-specific, 
time specific, sector-time, and state-sector specific) are not enough 
to control for state-time specific shocks
such as an unusually cold winter or a recession in a state,
that move both the other sectors' prices and sector 
$j$'s demand. \citet{BurkeAbayasekara2018} address 
it with state-level weather and income controls,
which we omit. 
\citet{BurkeAbayasekara2018} specify the 
first stage as linear, $g(z,\gamma_j,\eta_t)=\pi z+\kappa_{jt}$; we leave $g$ unrestricted.

First differences remove $\alpha_{sj}$ and
$\alpha^x_{sj}$ and give the equation of 
\citet{BurkeAbayasekara2018}, Table~5, column~6:
\[
\Delta\ln E_{sjt}=\beta_1\,\Delta\ln P_{sjt}+\lambda_{sj}+\Delta\mu_{jt}+\Delta\varepsilon_{sjt},
\]
where $\lambda_{sj}$ are state-by-sector effects and
$\Delta\mu_{jt}$ sector-by-year effects. Observations
are weighted by the square root of the number of customers, 
as in the original study.

\paragraph{Mapping to the model of Section~3.} The unit is 
$i=(s,j)$, with $N=144$ units, and $t$ is the year, with
12 first differences per unit ($T=13$); $x_{1it}=\ln P_{sjt}$,
$z_{it}=\ln P^{\mathrm{other}}_{sjt}$. The sector-by-year
effects are 36 indicators, which enter $\tilde x_{it}$ as a
fixed number of exogenous regressors; with $T$ fixed, 
estimating them does not affect the asymptotic distribution. 
The state-by-sector effects $\lambda_{sj}$ enter the 
differenced equation as 144 unit indicators, estimated 
together with the other coefficients; they are not removed 
by a transformation for simplicity but
they can be removed using an appropriate 
transformation as the theory presented in this paper is general
for any transformation that preserves independence 
in the cross-sectional dimension.  The robustness paragraph 
reports the
 estimates without $\lambda_{sj}$. We assume that unobserved idiosyncratic
 errors are 
 independent after conditioning on the covariates, and included fixed effects,
 we also report standard errors clustered by state (48 clusters).

\paragraph{Estimators.} We report first-differenced OLS and 2SLS with their
 linear instrument, with $V_i=I$, which reproduce the specification of 
 \citet{BurkeAbayasekara2018}. Our estimator is FD orthogonal SLCF with 
 $V_i=D_iD_i'$, with all variables
  are premultiplied, unit by unit, by $R_i$ the upper triangular matrix of the
  Cholesky decomposition
  $R_i'R_i=(D_iD_i')^{-1}$.
   By Proposition~\ref{Proposition_Comparison} it is IV with the learned instrument. 
   With $V_i=D_iD_i'$, we use FD-GLS and FD-G2SLS.
Since $D_i'(D_iD_i')^{-1}D_i=I_T-\iota\iota'/T$, these equal the within
OLS and within 2SLS estimators of the equation in levels, in which the
$\lambda_{sj}$ become state-by-sector linear trends. The three estimators
in this column use the same weight, so they differ only in the
instrument: none in the case of plain-vanilla least-squares estimation,
 $z$ in the case of 2SLS, or the learned instrument in the case of the orthogonal
 super learner control function estimator.

\paragraph{First stage.} With first differences,
$\Delta\ln P_{sjt}=g(z_{sjt},\gamma_j,\eta_t)-g(z_{sj,t-1},
\gamma_j,\eta_{t-1})+\Delta u_{sjt}$, so the super learner 
estimates $\E[\Delta\ln P_{sjt}\mid z_{sjt},z_{sj,t-1},\gamma_j,
\eta_t]$. Its library contains the mean, linear regression and 
gradient boosting (learning rate 0.1, depth 6, 300 rounds); its
weights are chosen by non-negative least squares on 3-fold
cross-validated predictions, with folds formed by state.
Functions of the sector and year alone are absorbed by the 
sector-by-year effects, so the learned instrument varies only
through $z$ and its interactions with sector and year.

\paragraph{Cross-fitting and inference.} The instrument of one 
sector is built from the prices of the other sectors of the same
state, so we form the $B=5$ cross-fitting folds by state. We solve
the moment conditions once over all folds (the pooled version of
Algorithm~\ref{Algortithm_CF1_ex}) and repeat the procedure for
$SS=10$ sample splits. The variance is estimated using the proposed 
variance-covariance estimator (\ref{Var_estimator}, \ref{J}) with 
$V_i=D_iD_i'$, pooled over folds, without a small-sample factor
(Remark \ref{Within_correction}), combined over splits by the median:
$\hat V_1 = \operatorname{median}\{\hat V_{ss} + (\hat\beta_{1,ss}-\hat\beta_1)^2 : ss = 1,\dots,SS\}$
(Subsection~\ref{Including_sample_splitting_uncertainty}).
In (\ref{Var_estimator}), the outer products are summed by unit; for
the state-clustered standard errors they are summed by 
state. The Anderson--Rubin set, and the first-stage $F$
statistic are standard 
cluster-robust statistics computed on the transformed data.

\paragraph{Data.}\label{S_EA_data} Sales, revenue and customers by
state and sector come from the historical EIA-861 files of the U.S.
Energy Information Administration \citep{EIA861hist}, the source of 
\citet{BurkeAbayasekara2018}. The price paid by the other sectors is 
total revenue minus the sector's revenue, divided by total sales 
minus the sector's sales, where totals include the transportation 
sector. Alaska, Hawaii and the District of Columbia are excluded. 
The sample is 2004--2015 (first differences of 2003--2015 levels):
1,728 observations, as in the original study. Our specification
differs from theirs in three ways: it omits their state-level 
controls (changes in log GDP per capita, population, heating and
cooling degree days and the natural gas price); prices are nominal
(a national deflator is absorbed by the sector-by-year effects, 
their regional deflators are not); and we use a 2026 release of 
the EIA data, which includes later revisions.

\subsection{Results}

\begin{table}[htbp]\centering\small
\caption{Short-run price elasticity of electricity demand, US states, 2004--2015}\label{tab:ea_burke}
\setlength{\tabcolsep}{4pt}
\begin{tabular}{lccccccc}\toprule
& \multicolumn{3}{c}{$V_i=I$} & \multicolumn{3}{c}{$V_i=D_iD_i'$} & Published\\ \cmidrule(lr){2-4}\cmidrule(lr){5-7}
& FD-OLS & FD-2SLS & \shortstack{FD orth.\\ SLCF} & FD-GLS & FD-G2SLS & \shortstack{FD orth.\\ SLCF} & 2SLS \\ \midrule
$\hat\beta_1$ & $-0.099$ & $-0.045$ & $-0.046$ & $-0.113$ & $-0.099$ & $-0.099$ & $-0.06$ \\
SE, clustered by unit & (0.022) & (0.027) & (0.027) & (0.037) & (0.047) & (0.047) & (0.02) \\
SE, clustered by state & & (0.023) & (0.023) & & (0.031) & (0.031) & \\
AR 95\% set & & & {\scriptsize$[-0.101,0.011]$} & & & {\scriptsize$[-0.194,0.000]$} & \\
$\hat\rho$ & & & $-0.131$ & & & $-0.022$ & \\ \midrule
First-stage partial $R^2$ & & 0.669 & 0.667 & & 0.758 & 0.763 & \\
First-stage $F$ & & 861.2 & 732.3 & & 578.9 & 667.7 & 2,420.9 \\
Out-of-fold $R^2$ & & & 0.649 & & & 0.649 & \\
\bottomrule\end{tabular}
\par\smallskip\parbox{0.97\linewidth}{\footnotesize Outcome:
$\Delta\ln$ electricity sales; endogenous regressor:
$\Delta\ln$ average price. 1,728 observations, 144 units
(48 states $\times$ 3 sectors), 12 first differences per unit. 
All columns include sector-by-year and state-by-sector effects
in the differenced equation and are weighted by the square root 
of the number of customers. Instrument: $\Delta\ln$ average 
price paid by the other sectors of the same state. FD orthogonal
SLCF: cross-fitting folds by state, median over 10 sample
splits; variance (38)--(40); super learner weights (mean, 
linear, boosting) 0.00/0.74/0.26 with both weights; $\hat\rho$ is reported
as a descriptive statistic. FD-OLS, FD-2SLS,  FD-GLS, FD-G2SLS: 
cluster-robust variance with factor $G/(G-1)$.
AR set and $F$ statistics clustered by unit.
Published: 
\citet{BurkeAbayasekara2018}, Table~5, 
column~6, with state-level controls and real prices.}
\end{table}

\paragraph{Replication.} With $V_i=I$, FD-OLS gives $-0.099$
(0.022), as their $-0.10$ (0.03) without instruments and controls 
(their Table~5, column~1), and FD-2SLS gives $-0.045$ (0.027),
against their $-0.06$ (0.02). The difference is within one standard
 error; our specification does not include their controls.

\paragraph{First stage.} The first stage is strong. With
$V_i=D_iD_i'$ the partial $R^2$ of the linear instrument is 0.758 
and that of the learned instrument 0.763. The super learner puts a 
weight of 0.74 on linear regression and 0.26 on boosting: in first 
differences of logs, the change in a sector's price is close to 
linear in the change in the price paid by the other sectors of the
same state.

\paragraph{Estimates.} With $V_i=D_iD_i'$, FD orthogonal SLCF gives
$-0.099$ (0.047), the same as FD-G2SLS with a linear specification. The Anderson--Rubin 95\%
set is $[-0.194,0.000]$. The estimate is twice the FD-2SLS estimate 
with $V_i=I$: the weight $V_i=D_iD_i'$ uses the variation in levels 
within a unit rather than only the year-to-year changes, so the two
estimators weight the data differently; the difference, 0.054, is
about one standard error of the estimate with $V_i=D_iD_i'$. It is 
close to the short-run elasticity of about $-0.1$ that 
\citet{BurkeAbayasekara2018} report as their summary across 
specifications. Standard errors clustered by state are smaller 
than those clustered by unit (0.031 against 0.047), which indicates
 negative correlation of the scores across the sectors of a state; 
 this is plausible because the instrument of one sector is built
  from the prices of the others.

\paragraph{Robustness.} Without the state-by-sector effects 
in the differenced equation, FD-G2SLS with $V_i=D_iD_i'$ gives 
$-0.096$ (0.057) and FD orthogonal SLCF $-0.079$ (0.058). 
 With $V_i=I_i$ and without these effects, FD-2SLS gives $-0.049$ 
 (0.027) and FD orthogonal SLCF $-0.048$ (0.028).

\paragraph{Summary.} In this published design the linear instrument
is strong and the first stage is close to linear. The super learner
gives most weight to the linear model, and FD orthogonal SLCF
is close to FD-G2SLS computed with the same $V_i$, in its estimate
and its standard error. The application shows that the estimator
and its inference work in a real panel; it does not test the 
weak-instrument case, for which real data provide no benchmark.
The gains from learning a nonlinear first stage are shown in 
Section~\ref{S_Simulation}.

\section{Conclusions}\label{S_Conclusions}

This paper proposes a triangular simultaneous equation model for panel data 
with additive separable individual-specific fixed effects composed of a 
linear structural equation with a nonlinear reduced form equation to deal with
 the issue of weak instrumental variables due to nonlinearities in panel data 
 settings. The parameter of interest is the structural parameter of the endogenous
  variable.  The identification of this parameter is obtained under the assumption
   of available exclusion restrictions and using a control function approach. 
    Estimating the parameter of interest is done using an estimator that we call the orthogonal
     Super Learner Control Function (Orthogonal SLCF) estimator.  The estimation procedure 
     is composed of two steps and cross-fitting.  
     We estimate the control function with a super learner using sample 
     splitting.  In the following step, we use the estimated control function
      to control for endogeneity in the structural equation. 
       Cross-fitting is done across the individual dimension.  
        The estimator is consistent and asymptotically normal achieving a
         parametric rate of convergence. We show that the estimator of the
         $\beta_o$ coefficients is numerically
         equivalent to the IV estimator using as instruments the transformed 
         nuisance parameter and the covariates.
        We perform a Monte Carlo simulation to test the performance of
        the estimators proposed.  We conclude that the estimator performs 
        well, provided that we can accurately learn  the nuisance parameter
         in the first stage.

\section*{Use of generative AI}
During the preparation of this work I used Claude (Anthropic) to
review the manuscript, check derivations, assist with
code for the simulations and the empirical application, and improve 
the wording and
notation. I reviewed and verified all output, including every
reference, and take full responsibility for the content of the paper.

\section{Appendix}\label{Appendix}

\subsection{Notation}
A set $\{1, 2, ..., N \}$ is denoted as $[N]$. $|\cdot|$ refers to the cardinality of a set.  $\left \|\cdot \right \|_{L^p}$ is the $L^p$ norm defined on the space $L^p(\Omega, P)$. $\left \|\cdot \right \|_p$ is the p-norm.
\subsection{Definitions}
We follow \cite{EmmeneggerSJS2023}  and provide the following definitions.  
 Each individual $i$ presents $T_i$ observations. We assume that $T_i$ is uniformly bounded such that $T_i < T_{max}$. The total number of observations is denoted by $N_T = \sum_i T_i$.  The number of sample splits $B$ is not random and independent of $N$. The data is defined as $\{ O_i = (X_i, z_i, y_i) \}_{i \in [N]}$ with $X_i=[x_{1i}, \tilde{X}_{i}]$, and denote the probability distribution $P$ of the grouped data $O_i$.  $\{ \mathcal{P}_N\}_{N\geq 1}$ is a sequence of probability distributions. 
A partition of the set $[N]$ is $S_1$, $S_2$, ..., $S_B$.  The total number of observations belonging to a partition $S_b \in [N]$ is equal to $n_{T,b} = \sum_{i \in S_b} T_i$.  We assume that the subsets $S_b$ have similar size such that $B n_{T,b} = N_T + o(1), \quad \forall b \in [B]$.  We denote $S_b^c = \{ O_i = (y_i, X_i, z_i)\}_{i \in S_b^c}$, and the estimator of the nuisance parameter $\widehat{\tau g}_o^{S_b^c}$ with data from $S_b^c$.
The empirical mean $\mathbb{E}_{n_{T,b}} = \frac{1}{n_{T_b}} \sum_{i\in S_b} \phi(O_i; \theta, \tau g)$. 
We define $\{\delta_N\}_{N \geq B}$ as a sequence of positive numbers that converge to 0 as $N \rightarrow \infty$, and $\{\delta_N\} \geq N^{-1/4}$.  
\subsubsection{Population Loss Function}

We define the following population loss function:

$$Q(\theta, \tau g_o) = \mathbb{E} [(M_{\tau_i} (y_i - H_i \theta))' V_i ^{-1} M_{\tau_i} (y_i - H_i \theta)],$$ 

\noindent where $ V_i$ is equal to the identity matrix $I_i$ under the within transformation, 
and $\Sigma = \mathbb{E} [(M_{\tau_i} \omega_i)(M_{\tau_i} \omega_i) '| H_i]$ under first-difference
transformation, also notice that $V_i$ is a scalar if $T_i=2, \forall i \in [N]$, 
$M_{\tau_i} = D_i$, and the error term $\omega_{it}$ is homoskedastic and
 uncorrelated across $i$ and $t$. Finally, $\theta = [\beta_1, \quad \beta_2', \quad \rho]'$

\subsubsection{Empirical Loss Function}

We define the empirical loss function:

$$L_N(\theta, \tau g_o) = \sum_i^N \frac{1}{N}[(M_{\tau_i}y_i - M_{\tau_i}H_{i}\theta )' V_i ^{-1} (M_{\tau_i}y_i -M_{\tau_i} H_i \theta)].$$

\subsubsection{Population moment condition}

We define the population moment condition:

$$m(O_i; \theta, \tau g_o) = \mathbb{E} [(M_{\tau_i}y_i - M_{\tau_i}H_{i}\theta )'V_i^{-1} M_{\tau_i}H_i ]$$ 
with $m(O_i; \theta_o, \tau g_o)=0$.

\subsubsection{Non Neyman orthogonal score function }
The score function is $\phi(O_i; \theta,\tau g_o)=(M_{\tau_i}y_i - M_{\tau_i}H_{i}\theta )' V_i^{-1}M_{\tau_i}H_i $.

\subsubsection{Sample moment condition}

We define the sample moment condition:

$$\hat{m}(O_i; \theta, \tau g_o) = \frac{1}{N}\sum_i (M_{\tau_i}y_i - M_{\tau_i}H_{i}\theta )'V_i^{-1} M_{\tau_i}H_i .$$ 
\subsection{Assumptions}\label{Annex_assumptions}
\begin{assumption}\label{A_ThetaCompact} $\theta \in \Theta$, the set $\Theta$ is bounded, contains $\theta_o$. 
\end{assumption}
This assumption is  similar to assumption 4.1 of \cite{EmmeneggerSJS2023}.

\begin{assumption} \label{A_nuisance}
The set $\mathcal{G}$ consists of P-integrable functions
 $\tau g$ with pth moment that exists and it
  contains $\tau g_o$,  
  $$\left\| \tau g - \tau g_o \right\|_{P,2} \leq \delta_N,$$
    $$\left\| \tau g - \tau g_o \right\|^2_{P,2} \leq N_T^{-1/2}\delta_N.$$      
\end{assumption}

\begin{assumption}\label{A_several} For all N, all $i \in [N]$, all $P \in \mathcal{P}_N$, all $b \in [B]$, and $p >8$ we have: 
\begin{enumerate}[label=\theassumption.\arabic*]
    \item At the true \( \theta_0 \) and the true \( g_0 \), the data \( \{ O_i = ( X_i, Z_i, Y_i) \}_{i \in [N]} \) satisfies the identifiability condition.
    $$\mathbb{E}_P \Big[\mathbb{E}_{n_{T,b}} [\phi(O_i;\theta_o, \tau g_o)]\Big]=0.$$ \label{A_identification}

    \item $\mathbb{E}_P \Big[\mathbb{E}_{n_{T,b}}[\phi(O_i;\theta_o, \tau g_o)
    \phi(O_i;\theta_o, \tau g_o)'] \Big]= \tilde{\Omega}_i$ with $\tilde{\Omega}_i$ a positive
    semi-definite matrix, with eigenvalues bounded away from 0 and 
     finite, with the minimum eigenvalue $\lambda_T>0$, and the maximum 
     eigenvalue $ \lambda_1<\infty$. \label{A_varcov_moments}

     \item $\mathbb{E}_P \Big[\mathbb{E}_{n_{T,b}}[\psi(O_i;\theta_o, \tau g_o)
    \psi(O_i;\theta_o, \tau g_o)'] \Big]= \Omega_i$ with $\Omega_i$ a positive
    semi-definite matrix, with eigenvalues bounded away from 0 and 
     finite, with the minimum eigenvalue $\tilde{\lambda}_T>0$, and the maximum 
     eigenvalue $\tilde{\lambda}_1<\infty$. \label{A_varcov_moments}

     \item $V_i$ is a positive semi-definite matrix, with eigenvalues bounded away
      from 0 and finite, with the minimum eigenvalue $\overline\lambda_T>0$, and the maximum 
     eigenvalue $\overline\lambda_1<\infty$. \label{A_varcov_moments_Vi}
    
    \item There exists finite real constants $C_1$,$C_2$, $C_3$, $C_4$ satisfying $|| y_i||_{P,p} \leq C_1$, $|| z_i||_{P,p} \leq C_2$, $|| x_{1,i}||_{P,p} \leq C_3$, $|| \tilde{x}_{i}||_{P,p} \leq C_4$ for all $k >2$.\label{A_bound}

\end{enumerate}
  
\end{assumption}

 \begin{assumption} \label{A_centered}
    $\mathbf E[\tau x_{k,i}] = 0$ for all $i \in [N]$, and all $k \in [K]$. 
    
 \end{assumption}

Assumption \ref{A_identification} states that the true parameter vector is identifiable. 
Assumption \ref{A_bound} states that the data has finite moments. 
This assumption is equivalent to Assumption 5.1 of \cite{EmmeneggerSJS2023}, and \cite{Newey1994Chapter}.

Assumption \ref{A_centered} states that the unconditional expectation of the
 transformed regressors is equal to 0.  This is plausible 
if the covariates are stationary.

\subsection{Proof of Proposition Identification}

For the proof, we follow \cite{chen2014local}.

The first-derivative of $\mathbb{E}[\phi(O_i;\theta,\tau g_o)]$ with respect 
$\theta_o$ is $\mathbb{E}[\phi'(O_i;\theta_o,\tau g_o)] =-\mathbb{E}[(M_{\tau_i} H_{i})'
    V_i^{-1}M_{\tau_i} H_{i}]$. Since the first-derivative of 
    the moment $\mathbb{E}[\phi(O_i;\theta,\tau g_o)]$ 
has full rank, then $\mathbb{E}[\phi'(O_i;\theta_o,\tau g_o)'\phi'(O_i;\theta_o,\tau g_o)]$  
has smallest eigenvalue $\lambda_{min}$ is positive. Also, 
$|\mathbb{E}[\phi'(O_i;\theta_o,\tau g_o)] '\ell| \geq \lambda_{min}\ell$. Then, $\exists \zeta$ such that 

\begin{equation}
\frac{|\mathbb{E}[\phi(O_i;\theta,\tau g_o)] - \mathbb{E}[\phi(O_i;\theta_o,\tau g_o)] -
\mathbb{E}[\phi'(O_i;\theta_o,\tau g_o)(\theta - \theta_o)]|}
{|\mathbb{E}[\phi'(O_i;\theta_o,\tau g_o)](\theta - \theta_o)|} < 1
\end{equation}

Thus, $\mathbb{E}[\phi(O_i;\theta,\tau g_o)] \neq 0$ for any $\theta \neq \theta_o$.
\subsection{Proof of Proposition Concentration parameter inequality }

\textbf{Step 1} By the Projection Theorem, we can re-write the function 
$\tau g_o(\tilde{x}_{it}, z_{it})$ as $[\tau \tilde{x}_{it} \quad \tau z_{it}]' a 
+ \tau \tilde{g}_o(\tilde{x}_{it}, z_{it})$.   \\

\noindent \textbf{Step 2} By the Pythagorean Theorem, we have that:

\begin{equation}
|| [  M_{\tau_i}z_{i} \quad  M_{\tau_i} \tilde{x}_{i} ] a||_{\Phi_i^+,P,2} \leq 
 || M_{\tau_i} g_o(\tilde{x}_{i}, z_{i}) ||_{\Phi_i^+,P,2},
\end{equation}

\noindent with $\Phi_i = \mathbb{E}[M_{\tau_i}u_iu_i'M_{\tau_i}']$.

\noindent \textbf{Step 3} Define the error term $\tau\tilde{u}_{it} = \tau x_{1it} 
- [\tau\tilde{x}_{it} \quad \tau z_{it} ]' a$.

 Notice that 
 \begin{equation}
    \begin{split}
    \tau\tilde{u}_{it} 
    & =  \tau x_{1it}  - [\tau\tilde{x}_{it} \quad \tau z_{it} ]' a \\
    & =  \tau g_o(x_{it}, z_{it}) + \tau u_{it} - [\tau\tilde{x}_{it} \quad \tau z_{it} ]' a \\
    & = \tau \tilde{g}_o(x_{it}, z_{it}) + \tau u_{it} 
    \end{split}
\end{equation}

Thus, the variance of the error term $\tau\tilde{u}_{it}$ is larger than or equal to the variance of 
 $\tau u_{it} = \tau x_{1it}  - \tau g_o(x_{it}, z_{it})$ since 
 $\mathbb{E} [\tau \tilde{g}_o(x_{it}, z_{it}) \tau u_{it} ]=0$ by L.I.E. and Assumption
 \ref{A_reduced_form}.  Stacking up the time observations, we have that: 
 
 \begin{equation}
 \begin{split}
    \mathbb{E}[(\tau \tilde{g}_o(\tilde{X}_i, z_i) + \tau u_i)
    (\tau \tilde{g}_o(\tilde{X}_i, z_i) + \tau u_i)'] = & \\
    \mathbb{E}[\tau \tilde{g}_o(\tilde{X}_i, z_i) \tau \tilde{g}_o(\tilde{X}_i, z_i)' + \tau u_i)
   + 2 \tau \tilde{g}_o(\tilde{X}_i, z_i) \tau u_i' + \tau u_i \tau u_i' ] = & \\
   \mathbb{E}[\tau \tilde{g}_o(\tilde{X}_i, z_i) \tau \tilde{g}_o(\tilde{X}_i, z_i)']+
   \mathbb{E}[\tau u_i \tau u_i'],  
 \end{split}
\end{equation}
\noindent where for ease of notation, $M_{\tau_i}$ is represented by $\tau$, and
 the equality holds since 
the $\mathbb{E}[\tau \tilde{g}_o(\tilde{X}_i, z_i) \tau u_i']=0$ by Assumption 
\ref{A_reduced_form}. 

Then, $\forall \mathbf{l} \in \mathbb{R}^{T_{ai}}$, 

\begin{equation}
\mathbf{l}'\Big[\mathbb{E}[\tau \tilde{g}_o(\tilde{X}_i, z_i) \tau \tilde{g}_o(\tilde{X}_i, z_i)'+
\tau u_i \tau u_i'] - \mathbb{E}[\tau u_i \tau u_i']  \Big] \mathbf{l} \geq 0 
\end{equation}

\textbf{Step 4} By step 3, we partial out $M_{\tau_i}\tilde{X}_i$ using 
$\tilde{\Gamma}=
\big(\mathbb{E}[\tilde{X}_i'M_{\tau_i}'\Phi_{\tau \tilde{u},i}^{+}M_{\tau_i}\tilde{X}_i]\big)^{-1}
\big(\mathbb{E}[\tilde{X}_i'M_{\tau_i}'\Phi_{\tau \tilde{u},i}^{+}M_{\tau_i}Z_i a_1]\big)$, 
and $\Gamma=
\big(\mathbb{E}[\tilde{X}_i'M_{\tau_i}'\Phi_{\tau u,i}^{+}M_{\tau_i}\tilde{X}_i]\big)^{-1}
\big(\mathbb{E}[\tilde{X}_i'M_{\tau_i}'\Phi_{\tau u,i}^{+}M_{\tau_i}Z_i a_1]\big)$.
Now, the following inequality holds because $\tilde{\Gamma}$ 
is the minimizer of the objective function.
    \begin{equation}
\mathbb{E} \Big[ (M_{\tau_i}Z_i a_1 - M_{\tau_i}\tilde{X}_i\tilde{\Gamma})'
\Phi_{\tau \tilde{u},i}^{+}
 (M_{\tau_i}Z_i a_1 - M_{\tau_i}\tilde{X}_i\tilde{\Gamma}) \Big] 
 \leq \mathbb{E} \Big[ (M_{\tau_i}Z_i a_1 - M_{\tau_i}\tilde{X}_i\Gamma)'
 \Phi_{\tau \tilde{u},i}^{+}
 (M_{\tau_i}Z_i a_1 - M_{\tau_i}\tilde{X}_i\Gamma) \Big].  
\end{equation}

Now, the following inequality holds because $\mathbf{l}'(\Phi_{\tau u,i}^{+}-\Phi_{\tau \tilde{u},i}^{+})\mathbf{l}>0$,
\begin{equation}
\mathbb{E} \Big[ (M_{\tau_i}Z_i a_1 - M_{\tau_i}\tilde{X}_i\Gamma)'
\Phi_{\tau \tilde{u},i}^{+}
 (M_{\tau_i}Z_i a_1 - M_{\tau_i}\tilde{X}_i\Gamma) \Big] 
 \leq \mathbb{E} \Big[ (M_{\tau_i}Z_i a_1 - M_{\tau_i}\tilde{X}_i\Gamma)'
 \Phi_{\tau u,i}^{+}
 (M_{\tau_i}Z_i a_1 - M_{\tau_i}\tilde{X}_i\Gamma) \Big].  
\end{equation}

And by Step 2 we have that 

    \begin{equation}
\mathbb{E} \Big[ (M_{\tau_i}Z_i a_1 - M_{\tau_i}\tilde{X}_i\tilde{\Gamma})'
\Phi_{\tau \tilde{u},i}^{+}
 (M_{\tau_i}Z_i a_1 - M_{\tau_i}\tilde{X}_i\tilde{\Gamma}) \Big] 
 \leq \pi_o.  
\end{equation}

The equality holds when  
 $g_o(\tilde{x}_{it}, z_{it}) \in span(\tilde{x}_{it}, z_{it})$.

\subsection{Consistency}\label{Consistency_proof}
The proof of consistency of $\hat{\theta}_o$ follows \cite{EmmeneggerSJS2023}, and \cite{Newey1994Chapter}.

\begin{proof} {Proof of Theorem \ref{T_Consistency}} \\
    Condition i) is satisfied by Assumptions \ref{A_identification}, and \ref{A_varcov_moments}, condition ii) is satisfied by Assumption \ref{A_ThetaCompact}, condition iii) and iv) are satisfied by lemma 12, lemma 15 of \cite{EmmeneggerSJS2023}, and lemma 1.  This implies that the estimator $\hat{\theta}$ is consistent \citep{EmmeneggerSJS2023}. 
\end{proof}

\begin{lemma}\label{lemma_1}
We prove that lemma 13 of \cite{EmmeneggerSJS2023} holds in our set up.
\end{lemma}
\begin{proof}
   \begin{equation}
   \begin{split}        
       \psi(O_i;\theta,\tau g)-\psi(O_i;\theta, \tau g_o) &= (M_{\tau_i}H_{i,g})'V_i^{-1}M_{\tau_i}(y_i-x_i\beta -\rho(x_{1i}-g(\tilde{X}_i, z_i))) \\
       & + (M_{\tau_i}[g(\tilde{X}_{i},z_{i}), 
    \tilde{X}_{i}, \textbf{0}])'V_i^{-1}M_{\tau_i} (x_{1i} - g(\tilde{X}_i, z_i)) \rho \\
    & - 
    (M_{\tau_i}H_{i,g_o})'V_i^{-1}M_{\tau_i}(y_i-x_i\beta 
    -\rho(x_{1i}-g_o(\tilde{X}_i, z_i)))\\
    & - (M_{\tau_i}[g_o(\tilde{X}_{i},z_{i}), 
    \tilde{X}_{i}, \textbf{0}])'V_i^{-1}M_{\tau_i} (x_{1i} - g_o(\tilde{X}_i, z_i)) \rho \\
\end{split}       
   \end{equation}
Since $g(\tilde{X}_i,z_i)=g_o(\tilde{X}_i,z_i) + \xi_i$, we are left with:
      \begin{equation}
   \begin{split}        
       \psi(O_i;\theta,\tau g)-\psi(O_i;\theta, \tau g_o) =&
       (M_{\tau_i}u_i)'V_i^{-1}M_{\tau_i}\xi_i\rho + 
       (M_{\tau_i}\omega_i)'V_i^{-1}M_{\tau_i}\xi_i 
\end{split}       
   \end{equation}

\noindent 
 With $P \in \{\mathcal{P}_N\}_{N\geq 1}$, then:
   \begin{equation}
   \begin{split}
      \operatorname*{sup}_{\theta \in \Theta, \tau g \in \mathcal{G}}
       \left \|\mathbb{E}_P[\mathbb{E}_{n_{T,b}}[(M_{\tau_i}(u_i\rho +\omega_i))'V_i^{-1}M_{\tau_i}\xi_i]] \right\| & \leq 
       \operatorname*{sup}_{\theta \in \Theta,\tau g
       \in \mathcal{G}}  \left\|\mathbb{E}_P[(M_{\tau_i}u_i)'V_i^{-1}M_{\tau_i}\xi_i\rho]\right\|+  \operatorname*{sup}_{\theta \in 
       \Theta, \tau g \in \mathcal{G}} 
       \left \|\mathbb{E}_P[(M_{\tau_i}\omega_i)'V_i^{-1}M_{\tau_i}\xi_i ]\right \|  
       \\
       &  \leq \operatorname*{sup}_{ \in \Theta,
       \tau g \in \mathcal{G}} \left \|\mathbb{E}_P[\sum_t 
       \lambda_t (q_t'M_{\tau_i} u_i)(q_t'M_{\tau_i}\xi_i)\rho]\right \|
       \\ 
       & + \operatorname*{sup}_{ \in \Theta, \tau g \in
       \mathcal{G}} \left \| \mathbb{E}_P[\sum_t \lambda_t
       (q_t'M_{\tau_i}\omega_i')(q_t'M_{\tau_i}\xi_i)]\right \| 
       \\ 
       &\leq 
       \operatorname*{sup}_{\theta \in \Theta,\tau  g
       \in \mathcal{G}} \left \|\mathbb{E}_P[\lambda_1 
       \sum_t(q_t'M_{\tau_i}\omega_i)(q_t'M_{\tau_i}\xi_i)]\right \|+
       \\
       & \operatorname*{sup}_{\theta \in \Theta, \tau g \in \mathcal{G}} 
       \left \|\lambda_1 \mathbb{E}_P \sum_t(q_t'M_{\tau_i}\xi_i
       (q_t'M_{\tau_i}u_i)\rho]\right \| 
       \\
        &\leq \operatorname*{sup}_{\theta \in \Theta, \tau g \in \mathcal{G}}
         \lambda_1\left\| \left \|M_{\tau_i} \omega_i\right \|_{P,2}
         \left \|M_{\tau_i} \xi_i\right \|_{P,2}\right\|+ 
         \\ 
         &\operatorname*{sup}_{\theta \in \Theta, \tau g \in \mathcal{G}}
          \lambda_1 \rho \left\|\left \|M_{\tau_i} u_i\right \|_{P,2}
         \left \|M_{\tau_i} \xi_i\right \|_{P,2}\right\|
   \end{split}       
   \end{equation}
\end{proof}
Since $\Theta$ is bounded, $\left \| \tau g - \tau g_o \right \|_{P,2} \leq \delta_N$, the last expression is bounded by $\delta_N$ by assumptions \ref{A_ThetaCompact}, and \ref{A_nuisance}.

\subsection{Asymptotic Normality}\label{Proof_AsymptoticNormality}
\begin{proof} 
In order to prove the asymptotic normality of the SLCF estimator of $\beta_{1o}$, we write the estimator as: 

\begin{equation}
  \hat{\beta}_{o} = \frac{1}{B}\sum_{b=1}^B 
  \Big( \sum_{i \in S_b} 
(\widehat{\tilde{Z}'}_i V_i^{-1}M_{\tau_i} X_i)'\Big)^{-1}  \Big( \sum_{i \in S_b} 
(\widehat{\tilde{Z}'}V_i^{-1} M_{\tau_i} y_i)'\Big),
\end{equation}

\noindent with $\widehat{\tilde{Z}}=[\widehat{g}(\tilde{X}_i,z_i), \tilde{X}_i]$.

This is equal to: 

\begin{equation}
   \sqrt{N_T} (\hat{\beta}_{o} - \beta_{o}) = \frac{1}{B}
   \sum_{b=1}^B \Big(\frac{1}{n_{T,b}}
   \sum_{i \in S^b} \widehat{\tilde{Z}'} V_i^{-1} M_{\tau_i} X_i)\Big)^{-1}
    \Big(\frac{\sqrt{N_T}}{n_{T,b}}
    \sum_{i \in S^b} \widehat{\tilde{Z}'}V_i^{-1} M_{\tau_i} \epsilon_i \Big).
\end{equation}

with
\[
\widehat{\tilde{Z}}'V_i^{-1}M_{\tau_i}\varepsilon_i
= \tilde Z_i'V_i^{-1}M_{\tau_i}\varepsilon_i + e_1\,\xi_i'V_i^{-1}M_{\tau_i}\varepsilon_i,
\qquad \xi_i:=\widehat{\tau g}_i^{S_b^c}-\tau g_{o,i},\quad e_1=(1,0,\dots,0)'.
\]
The remainder $R_b:=(NT)^{-1/2}\sum_{i\in S_b}\xi_i'V_i^{-1}M_{\tau_i}\varepsilon_i$ is $o_p(1)$ by sample splitting: conditional on the data in $S_b^c$, $\xi_i$ is a fixed function of $(\tilde X_i,z_i)$, so the summands are independent with mean zero (because $\mathbb{E}[M_{\tau_i}\varepsilon_i\mid\tilde X_i,z_i]=0$) and $\mathbb{E}[R_b^2\mid S_b^c]\le C\,\|\xi\|_{P,2}^2\to_p 0$. Hence $R_b=o_p(1)$ by Chernozhukov et al.\ (2018, Lemma 6.1).

   Then, we have that:

\begin{equation}
    \begin{split}      
    \frac{\sqrt{N_T}}{n_{T,b}}\sum_{i \in S^b}
    \widehat{\tilde{Z}'}V_i^{-1} M_{\tau_i} \epsilon_i) 
   = & \frac{\sqrt{N_T}}{n_{T,b}}\sum_{i \in S^b} \tilde Z_i'V_i^{-1} M_{\tau_i} \epsilon_i + o_p(1) \\
    \end{split}
\end{equation}

We define: 

\begin{equation}
S_o = \mathbb{E} [ \tilde Z_{i}' V_i^{-1} M_{\tau_i} \epsilon_i \epsilon_i' M_{\tau_i}' V_i^{-1} 
\tilde Z_i] ,
\end{equation}

 Then, by the Lindeberg-Feller Central Limit Theorem we have that:

\begin{equation}
     \frac{1}{\sqrt{N_T}}\sum_{i} \tilde Z_i' V_i^{-1} M_{\tau_i} \epsilon_i \xrightarrow{d} \mathcal{N}(0, S_o).
\end{equation}

Thus, 
\begin{equation}
    \sqrt{N_T} (\hat{\beta}_{o} - \beta_{o}) \xrightarrow{d}  \mathcal{N}(0, \Sigma_o),
\end{equation}

with $\Sigma_o = \mathbb{E}_P[\tilde Z_i'V_i^{-1}M_{\tau_i}X_i]^{-1}
 \mathbb{E}_{P} [ \tilde Z_{i}' V_i^{-1} M_{\tau_i} \epsilon_i \epsilon_i' M_{\tau_i}' V_i^{-1} 
\tilde Z_i] (\mathbb{E}_P[\tilde Z_i'V_i^{-1}M_{\tau_i}X_i] ^{-1})'$.
\end{proof}

\providecommand{\E}{\mathbb{E}}

\providecommand{\simfigdir}{figures_sim/}

\section{Additional simulation results}\label{A_MC}

This appendix complements Section~\ref{S_Simulation}. All results use the design and implementation of that section: $N=2500$, $T=2$, first differences, $R=500$ replications.

\paragraph{First-stage diagnostics.} Table~\ref{tab:mc_diag} 
reports, for each $a$, the median over replications of
 the cluster-robust first-stage $F$ statistics (Wald statistics
  divided by the number of excluded instruments) of $\Delta z$,
   of a degree-5 polynomial in $\Delta z$ and of $\widehat{\Delta g}$,
    the out-of-fold $R^2$ of the super learner and $\hat\rho$. 
    For $\widehat{\Delta g}$, the $F$ statistic treats 
    $\widehat{\Delta g}$ as a fixed regressor and is reported
    only as a descriptive measure of instrument strength. 
    The $F$ statistic of $\Delta z$ is below 10 for $a\geq6$, 
    that of the polynomial for $a\geq4$, and that of
     $\widehat{\Delta g}$ is at least $2928$ and increases with $a$.

\paragraph{Why polynomials of $\Delta z$ do not help.} The $F$ 
statistic of the polynomial is smaller than that of $\Delta z$ 
for every $a\leq 9$. Given $\Delta z_{it}=d$, the map $(z_{it},z_{it-1})\mapsto(-z_{it-1},-z_{it})$ leaves $\Delta z_{it}$ and the joint distribution unchanged, because $z_{it}$ and $z_{it-1}$ are i.i.d.\ and symmetric around zero, and it changes the sign of $|z_{it}|-|z_{it-1}|$. Hence $\E\big[a(|z_{it}|-|z_{it-1}|)\mid\Delta z_{it}\big]=0$, and any function of $\Delta z_{it}$ is correlated with $\Delta g_o$ only through the linear term $\Delta z_{it}/a$. The higher-order terms add instruments without adding information. The nonlinear part of $\Delta g_o$ is a function of $(z_{it},z_{it-1})$, which the super learner uses as features.

\paragraph{FD-OLS.} The bias of FD-OLS falls with $a$ because the variance of $\Delta x_1$ explained by $g_o$ grows with $a$, while the variance of $\Delta u$ does not change. It is the same at $a=1$ and $a=2$ because $\operatorname{Var}(a|z_{it}|+z_{it}/a)=a^2/3+4/(3a^2)$ takes the value $5/3$ at both.

\paragraph{Dispersion and interval length.} Tables~\ref{tab:mc_iqr} and~\ref{tab:mc_len} report the interquartile range of $\hat\beta_1$ and the median length of the 95\% confidence intervals. For Naive SLCF, the intervals are too short relative to the sampling dispersion. Under normality, an interquartile range of $0.036$ (at $a=10$) corresponds to a standard deviation of $0.036/1.349\approx0.027$ and to a 95\% interval length of about $0.105$, whereas the median reported length is $0.035$.

\paragraph{Distribution of the Orthogonal SLCF estimator.} Figure~\ref{fig:mc_hist} shows the histogram of the 500 Orthogonal SLCF estimates of $\beta_{1_o}$ for each $a$, together with the normal density centred at $\beta_{1_o}=1$ whose variance is the square of the median reported standard error. The histograms are centred at the true value and close to the normal density for every $a$. The ratio of the Monte Carlo standard deviation of $\hat\beta_1$ to the median reported standard error is between $0.95$ and $1.06$. The sample skewness is between $-0.28$ and $0.09$ and the sample kurtosis between $2.79$ and $3.14$. With 500 replications, the standard errors of these statistics under normality are about $0.11$ and $0.22$; all values are within two standard errors of the normal values $0$ and $3$, except the skewness at $a=1$ ($-0.28$, about $2.6$ standard errors), a mild left skew.
The normal approximation of Theorem~\ref{Theorem_AN} and the variance estimator are therefore accurate in this design.

\paragraph{Estimate of $\rho_o$.} The median of $\hat\rho$ decreases from $0.836$ at $a=1$ to $0.695$ at $a=10$, against the true value $0.9$. The control is $\widehat{\Delta u}=\Delta u+(\Delta g_o-\widehat{\Delta g})$, so the first-stage error acts as measurement error in the control and attenuates $\hat\rho$. The scale of $g_o$, and with it the scale of the estimation error, increases with $a$, which is consistent with the attenuation increasing with $a$. Because the score of $\beta$ is orthogonal, this error does not affect $\hat\beta_1$ to first order.

\paragraph{Aggregation over folds and sample splits.} Table~\ref{tab:mc_rob} compares averaging over folds (Average of fold estimates) with pooling over folds (Pooled moment), and the median with the mean over sample splits. The three versions have median bias within $0.001$ of each other and coverage within $0.01$ of each other for every $a$.

\paragraph{Variance estimator.} The standard errors of Section~\ref{S_Simulation} use the variance estimator (38)--(39), which pools the Jacobian and the score variance over the $B$ folds. An alternative treats the $B$ fold estimates as independent and estimates the variance of their average by $B^{-2}\sum_b\hat V_b$, where $\hat V_b$ is the cluster-robust variance in fold $b$. Both are consistent. Table~\ref{tab:mc_var} reports the coverage with each. For Orthogonal SLCF the two differ by at most $0.002$, and the median ratio of the standard errors is between $1.000$ and $1.003$ for every $a$; for Naive SLCF and Plug-in 2SLS the coverage differs by at most $0.016$ and $0.004$.

\paragraph{Heterogeneous $\rho_i$.} We also consider $\rho_i=0.9+0.5\,
\alpha_i$. Since $\rho_i$ depends only on $\alpha_i$, which is 
independent of $u_i$, this design satisfies the condition
 $\E[\rho_iM_{\tau_i}u_i\mid\tilde x_i,z_i]=0$. 
 Table~\ref{tab:mc_het} reports the results for $a=1,5,10$. 
 They are the same as with homogeneous $\rho_o$ up to Monte Carlo
  error: Orthogonal SLCF has median bias of at most $0.0005$
   in absolute value and coverage between $0.94$ and $0.95$;
    Naive SLCF has coverage that falls from $0.72$ to $0.16$, 
    Plug-in 2SLS from $0.97$ to $0.18$, and Plug-in IV no cross-fit.
    \ has coverage between $0.07$ and $0.49$.

\begin{table}[htbp]\centering\small
\caption{First-stage diagnostics}\label{tab:mc_diag}
\begin{tabular}{rccccc}
\toprule
$a$ & $F$: $\Delta z$ & $F$: poly$(\Delta z)$ & $F$: $\widehat{\Delta g}$ & Out-of-fold $R^2$ & $\hat\rho$ \\
\midrule
1 & 1840.5 & 379.2 & 2927.7 & 0.701 & 0.836 \\
2 & 285.1 & 64.6 & 2950.7 & 0.702 & 0.838 \\
3 & 77.3 & 20.7 & 5353.9 & 0.768 & 0.824 \\
4 & 28.9 & 9.3 & 8876.1 & 0.826 & 0.810 \\
5 & 12.2 & 5.0 & 13350.0 & 0.867 & 0.791 \\
6 & 6.3 & 3.0 & 18773.1 & 0.897 & 0.772 \\
7 & 3.6 & 2.2 & 24723.5 & 0.918 & 0.744 \\
8 & 2.0 & 1.7 & 31413.9 & 0.932 & 0.731 \\
9 & 1.5 & 1.4 & 39260.2 & 0.944 & 0.715 \\
10 & 0.9 & 1.2 & 46938.4 & 0.952 & 0.695 \\
\bottomrule
\end{tabular}
\par\medskip\begin{minipage}{0.95\linewidth}\footnotesize Medians over 500 replications, $N=2500$, $T=2$, first differences. $F$: cluster-robust first-stage Wald statistic divided by the number of excluded instruments. Out-of-fold $R^2$ of the super learner for $\Delta x_{1}$. True $\rho_o=0.9$.\end{minipage}
\end{table}

\begin{table}[htbp]\centering\small
\caption{Interquartile range of $\hat\beta_1$}\label{tab:mc_iqr}
\setlength{\tabcolsep}{4pt}
\begin{tabular}{rcccccccc}
\toprule
$a$ & FD-OLS & FD-2SLS & \begin{tabular}[b]{@{}c@{}}FD-2SLS\\poly\end{tabular} & \begin{tabular}[b]{@{}c@{}}Oracle\\plug-in IV\end{tabular} & \begin{tabular}[b]{@{}c@{}}Orthogonal\\SLCF\end{tabular} & \begin{tabular}[b]{@{}c@{}}Naive\\SLCF\end{tabular} & \begin{tabular}[b]{@{}c@{}}Plug-in\\2SLS\end{tabular} & \begin{tabular}[b]{@{}c@{}}Plug-in IV\\no cross-fit.\end{tabular} \\
\midrule
1 & 0.014 & 0.026 & 0.027 & 0.022 & 0.024 & 0.034 & 0.036 & 0.032 \\
2 & 0.014 & 0.056 & 0.056 & 0.024 & 0.025 & 0.034 & 0.039 & 0.031 \\
3 & 0.012 & 0.089 & 0.088 & 0.017 & 0.017 & 0.037 & 0.040 & 0.021 \\
4 & 0.010 & 0.114 & 0.105 & 0.014 & 0.015 & 0.042 & 0.047 & 0.015 \\
5 & 0.010 & 0.151 & 0.124 & 0.012 & 0.012 & 0.044 & 0.048 & 0.014 \\
6 & 0.009 & 0.180 & 0.139 & 0.010 & 0.009 & 0.043 & 0.054 & 0.011 \\
7 & 0.007 & 0.196 & 0.146 & 0.008 & 0.008 & 0.042 & 0.049 & 0.009 \\
8 & 0.006 & 0.246 & 0.154 & 0.007 & 0.007 & 0.046 & 0.056 & 0.007 \\
9 & 0.006 & 0.251 & 0.134 & 0.006 & 0.006 & 0.041 & 0.052 & 0.007 \\
10 & 0.005 & 0.256 & 0.136 & 0.005 & 0.005 & 0.036 & 0.047 & 0.006 \\
\bottomrule
\end{tabular}
\par\medskip\begin{minipage}{0.95\linewidth}\footnotesize Interquartile range of $\hat\beta_1$ over 500 replications, $N=2500$, $T=2$, first differences.\end{minipage}
\end{table}

\begin{table}[htbp]\centering\small
\caption{Median length of 95\% confidence intervals}\label{tab:mc_len}
\setlength{\tabcolsep}{4pt}
\begin{tabular}{rcccccccc}
\toprule
$a$ & FD-OLS & FD-2SLS & \begin{tabular}[b]{@{}c@{}}FD-2SLS\\poly\end{tabular} & \begin{tabular}[b]{@{}c@{}}Oracle\\plug-in IV\end{tabular} & \begin{tabular}[b]{@{}c@{}}Orthogonal\\SLCF\end{tabular} & \begin{tabular}[b]{@{}c@{}}Naive\\SLCF\end{tabular} & \begin{tabular}[b]{@{}c@{}}Plug-in\\2SLS\end{tabular} & \begin{tabular}[b]{@{}c@{}}Plug-in IV\\no cross-fit.\end{tabular} \\
\midrule
1 & 0.040 & 0.080 & 0.080 & 0.069 & 0.072 & 0.049 & 0.147 & 0.061 \\
2 & 0.040 & 0.161 & 0.159 & 0.069 & 0.072 & 0.046 & 0.146 & 0.062 \\
3 & 0.035 & 0.243 & 0.237 & 0.052 & 0.053 & 0.044 & 0.111 & 0.048 \\
4 & 0.031 & 0.323 & 0.305 & 0.040 & 0.041 & 0.041 & 0.091 & 0.038 \\
5 & 0.027 & 0.411 & 0.367 & 0.032 & 0.033 & 0.042 & 0.080 & 0.031 \\
6 & 0.023 & 0.503 & 0.416 & 0.027 & 0.027 & 0.038 & 0.069 & 0.026 \\
7 & 0.021 & 0.596 & 0.446 & 0.023 & 0.023 & 0.047 & 0.075 & 0.022 \\
8 & 0.019 & 0.749 & 0.459 & 0.020 & 0.020 & 0.042 & 0.065 & 0.020 \\
9 & 0.017 & 0.825 & 0.448 & 0.018 & 0.018 & 0.038 & 0.061 & 0.018 \\
10 & 0.015 & 0.931 & 0.436 & 0.016 & 0.016 & 0.035 & 0.057 & 0.016 \\
\bottomrule
\end{tabular}
\par\medskip\begin{minipage}{0.95\linewidth}\footnotesize Median over 500 replications, $N=2500$, $T=2$, first differences.\end{minipage}
\end{table}

\begin{table}[htbp]\centering\small
\caption{Orthogonal SLCF: aggregation over folds and sample splits}\label{tab:mc_rob}
\setlength{\tabcolsep}{4pt}
\begin{tabular}{rcccccc}
\toprule
 & \multicolumn{3}{c}{Median bias} & \multicolumn{3}{c}{Coverage} \\
\cmidrule(lr){2-4}\cmidrule(lr){5-7}
$a$ & \begin{tabular}[b]{@{}c@{}}Average of\\fold estimates,\\median\end{tabular} & \begin{tabular}[b]{@{}c@{}}Pooled\\moment,\\median\end{tabular} & \begin{tabular}[b]{@{}c@{}}Average of\\fold estimates,\\mean\end{tabular} & \begin{tabular}[b]{@{}c@{}}Average of\\fold estimates,\\median\end{tabular} & \begin{tabular}[b]{@{}c@{}}Pooled\\moment,\\median\end{tabular} & \begin{tabular}[b]{@{}c@{}}Average of\\fold estimates,\\mean\end{tabular} \\
\midrule
1 & $-0.000$ & 0.000 & $-0.000$ & 0.94 & 0.94 & 0.94 \\
2 & $-0.001$ & $-0.001$ & $-0.002$ & 0.96 & 0.96 & 0.97 \\
3 & $-0.001$ & $-0.001$ & $-0.001$ & 0.95 & 0.96 & 0.96 \\
4 & $-0.002$ & $-0.002$ & $-0.002$ & 0.95 & 0.95 & 0.95 \\
5 & 0.000 & 0.000 & $-0.000$ & 0.93 & 0.94 & 0.94 \\
6 & 0.001 & 0.001 & 0.001 & 0.93 & 0.93 & 0.94 \\
7 & $-0.000$ & $-0.000$ & $-0.000$ & 0.96 & 0.96 & 0.96 \\
8 & $-0.000$ & $-0.000$ & $-0.000$ & 0.97 & 0.97 & 0.97 \\
9 & 0.000 & 0.000 & 0.000 & 0.94 & 0.94 & 0.94 \\
10 & $-0.000$ & $-0.000$ & $-0.000$ & 0.94 & 0.94 & 0.94 \\
\bottomrule
\end{tabular}
\par\medskip\begin{minipage}{0.95\linewidth}\footnotesize Average of fold estimates: fold-by-fold estimates averaged over folds. Pooled moment: one estimate pooled over folds. Median or mean over the 3 sample splits. The first column is the Orthogonal SLCF estimator of Table~\ref{tab:mc_main}. 500 replications, $N=2500$, $T=2$, first differences.\end{minipage}
\end{table}

\begin{table}[htbp]\centering\small
\caption{Coverage with two variance estimators for the average of fold estimates}\label{tab:mc_var}
\setlength{\tabcolsep}{5pt}
\begin{tabular}{rcccccc}
\toprule
 & \multicolumn{2}{c}{Orthogonal SLCF} & \multicolumn{2}{c}{Naive SLCF} & \multicolumn{2}{c}{Plug-in 2SLS}\\
\cmidrule(lr){2-3}\cmidrule(lr){4-5}\cmidrule(lr){6-7}
$a$ & (38)--(39) & fold sum & (38)--(39) & fold sum & (38)--(39) & fold sum \\
\midrule
1 & 0.940 & 0.940 & 0.726 & 0.716 & 0.964 & 0.962 \\
2 & 0.962 & 0.962 & 0.672 & 0.668 & 0.950 & 0.948 \\
3 & 0.954 & 0.954 & 0.686 & 0.680 & 0.932 & 0.930 \\
4 & 0.954 & 0.954 & 0.572 & 0.568 & 0.834 & 0.834 \\
5 & 0.934 & 0.932 & 0.516 & 0.504 & 0.760 & 0.758 \\
6 & 0.932 & 0.932 & 0.420 & 0.412 & 0.622 & 0.620 \\
7 & 0.956 & 0.956 & 0.362 & 0.348 & 0.548 & 0.544 \\
8 & 0.966 & 0.964 & 0.284 & 0.268 & 0.398 & 0.394 \\
9 & 0.936 & 0.934 & 0.214 & 0.204 & 0.260 & 0.256 \\
10 & 0.940 & 0.940 & 0.162 & 0.154 & 0.178 & 0.178 \\
\bottomrule
\end{tabular}
\par\medskip\begin{minipage}{0.95\linewidth}\footnotesize Coverage of nominal 95\% confidence intervals. (38)--(39): Jacobian and score variance pooled over the folds (used in Section~\ref{S_Simulation}). Fold sum: $B^{-2}\sum_b\hat V_b$. Same estimates; median over the $SS=3$ sample splits in both. 500 replications, $N=2500$, $T=2$, first differences.\end{minipage}
\end{table}

\begin{table}[htbp]\centering\small
\caption{Heterogeneous $\rho_i=0.9+0.5\,\alpha_i$: median bias and coverage}\label{tab:mc_het}
\setlength{\tabcolsep}{4pt}
\begin{tabular}{rcccccccc}
\toprule
$a$ & FD-OLS & FD-2SLS & \begin{tabular}[b]{@{}c@{}}FD-2SLS\\poly\end{tabular} & \begin{tabular}[b]{@{}c@{}}Oracle\\plug-in IV\end{tabular} & \begin{tabular}[b]{@{}c@{}}Orthogonal\\SLCF\end{tabular} & \begin{tabular}[b]{@{}c@{}}Naive\\SLCF\end{tabular} & \begin{tabular}[b]{@{}c@{}}Plug-in\\2SLS\end{tabular} & \begin{tabular}[b]{@{}c@{}}Plug-in IV\\no cross-fit.\end{tabular} \\
\midrule
\multicolumn{9}{l}{\textit{Median bias}}\\
1 & 0.382 & 0.002 & 0.003 & 0.001 & $-0.000$ & $-0.016$ & 0.019 & 0.068 \\
5 & 0.122 & 0.008 & 0.036 & 0.000 & 0.000 & $-0.028$ & 0.032 & 0.021 \\
10 & 0.034 & 0.021 & 0.033 & 0.000 & $-0.000$ & $-0.039$ & 0.051 & 0.008 \\
\midrule
\multicolumn{9}{l}{\textit{Coverage}}\\
1 & 0.00 & 0.93 & 0.93 & 0.93 & 0.95 & 0.72 & 0.97 & 0.07 \\
5 & 0.00 & 0.97 & 0.93 & 0.94 & 0.94 & 0.51 & 0.76 & 0.30 \\
10 & 0.00 & 1.00 & 0.99 & 0.95 & 0.95 & 0.16 & 0.18 & 0.49 \\
\bottomrule
\end{tabular}
\par\medskip\begin{minipage}{0.95\linewidth}\footnotesize 500 replications, $N=2500$, $T=2$, first differences. Median $\hat\rho$: $0.833$, $0.789$, $0.692$, and median out-of-fold $R^2$: $0.701$, $0.867$, $0.952$, for $a=1,5,10$.\end{minipage}
\end{table}

\begin{figure}[htbp]\centering
\includegraphics[width=\linewidth]{ 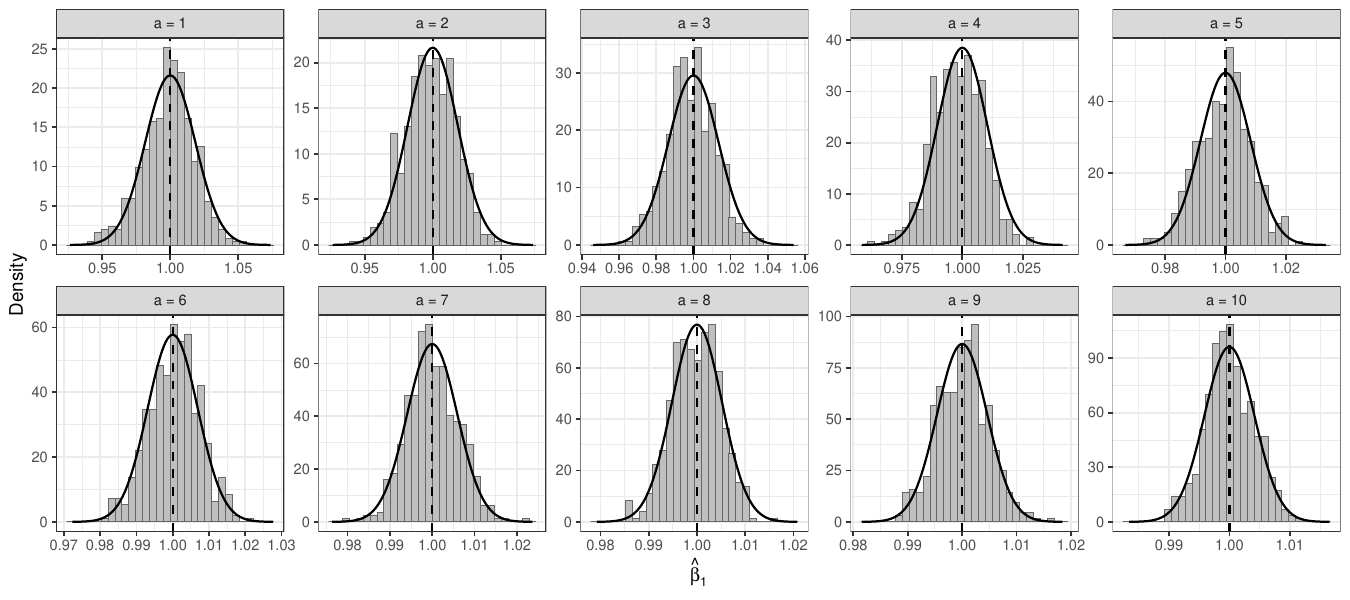}
\caption{Histograms of the Orthogonal SLCF estimates of $\beta_1$ by $a$. The solid line is the normal density with mean $\beta_{1_o}=1$ and standard deviation equal to the median reported standard error; the dashed line marks $\beta_{1_o}$. The horizontal scale differs across panels. 500 replications, $N=2500$, $T=2$, first differences.}\label{fig:mc_hist}
\end{figure}

\bibliographystyle{chicago} %
\bibliography{P4}  %

\end{document}